\documentclass{sigplanconf}
\usepackage{ifthen}
\newcommand{\techRep}{true} 
\newcommand{\iftechrep}{\ifthenelse{\equal{\techRep}{true}}}

\usepackage{amsthm} 
\newtheorem{theorem}{Theorem}
\newtheorem{lemma}[theorem]{Lemma}
\newtheorem{corollary}[theorem]{Corollary}
\newtheorem{definition}[theorem]{Definition}
\newtheorem{proposition}[theorem]{Proposition}

\newtheorem{remark}[theorem]{Remark}

\usepackage{amsmath,amssymb}
\usepackage{tikz}
\usetikzlibrary{arrows,calc,automata}
\tikzset{LMC style/.style={>=stealth',every edge/.append style={thick},every state/.style={minimum size=20,inner sep=0}}}

\RequirePackage{float}
\floatstyle{ruled}
\newfloat{algorithm}{tp}{lop}
\floatname{algorithm}{Algorithm}
\newcommand{\ind}{\mbox{}\hspace{4mm}}
\newcommand{\indd}{\ind \ind}
\newcommand{\inddd}{\indd \ind}

\newcommand{\co}{\mathit{co}}%
\newcommand{\con}{\mathit{con}}
\newcommand{\F}{\mathcal{F}}
\newcommand{\fs}[1]{f^{(#1)}}%
\newcommand{\li}{\mathit{li}}%
\newcommand{\M}{\mathcal{M}}

\newcommand{\mi}{\mathit{min}}
\newcommand{\N}{\mathbb{N}}
\renewcommand{\Pr}{\textup{Pr}}

\newcommand{\Q}{\mathbb{Q}}
\newcommand{\R}{\mathbb{R}}
\newcommand{\Run}{\mathit{Run}}
\newcommand{\supp}{\mathit{supp}}

\newenvironment{qtheorem}[1]{%
{\par\medskip\noindent\bf Theorem #1.}
\begin{itshape}%
}{%
\end{itshape}%
\smallskip
}
\newenvironment{qlemma}[1]{%
{\par\medskip\noindent\bf Lemma #1.}%
\begin{itshape}%
}{%
\end{itshape}%
\smallskip
}

\newenvironment{qproposition}[1]{%
{\par\medskip\noindent\bf Proposition #1.}
\begin{itshape}%
}{%
\end{itshape}%
\smallskip
}

\begin{document}

\setlength{\pdfpageheight}{\paperheight}
\setlength{\pdfpagewidth}{\paperwidth}

 \conferenceinfo{CSL-LICS 2014}{July 14--18, 2014, Vienna, Austria}
 \copyrightyear{2014}
 \copyrightdata{978-1-4503-2886-9}
 \doi{2603088.2603099} 

\title{On the Total Variation Distance of Labelled Markov Chains}
\subtitle{}

\authorinfo{Taolue Chen}
           {Middlesex University London}
           {t.l.chen@mdx.ac.uk}
\authorinfo{Stefan Kiefer}
           {University of Oxford}
           {stekie@cs.ox.ac.uk}

\maketitle

\begin{abstract}
 Labelled Markov chains (LMCs) are widely used in probabilistic verification, speech recognition, computational biology, and many other fields.
Checking two LMCs for equivalence is a classical problem subject to extensive studies,
while the total variation distance provides a natural measure for the ``inequivalence'' of two LMCs:
it is the maximum difference between probabilities that the LMCs assign to the same event.

In this paper we develop a theory of the total variation distance between two LMCs, with emphasis on the algorithmic aspects: (1) we provide a polynomial-time algorithm for determining whether two LMCs have distance 1, i.e., whether they can almost always be distinguished;
(2) we provide an algorithm for approximating the distance with arbitrary precision; and
(3) we show that the threshold problem, i.e., whether the distance exceeds a given threshold, is NP-hard and hard for the square-root-sum problem. We also make a connection between the total variation distance and \emph{Bernoulli convolutions}.
\end{abstract}

\category{G.3}{Probability and Statistics}{}
\category{D.2.4}{Software/Program Verification}{}

\terms
Theory

\keywords
Labelled Markov Chains, Total Variation Distance

\section{Introduction}

A (discrete-time, finite-state) \emph{labelled Markov chain (LMC)} has a finite set~$Q$ of states
 and for each state a probability distribution over its outgoing transitions.
Each outgoing transition is labelled with a letter from a given finite alphabet~$\Sigma$, and leads to a target state.
Figure~\ref{fig-intro-irrational} depicts two LMCs.
\begin{figure}[!h]
\begin{center}
\scalebox{0.95}{
\begin{tikzpicture}[scale=2.5,LMC style]
\node[state] (q1) at (-1,0) {$q_1$};
\node[state] (q2) at (+1,0) {$q_2$};
\node[state] (r1) at ( 0,0) {$r_1$};
\node[state] (r2) at (+2,0) {$r_2$};
\path[->] (q1) edge node[above] {$\frac14 c$} (r1);
\path[->] (q2) edge node[above] {$\frac14 c$} (r2);
\path[->] (q1) edge [loop,out=110,in=70,looseness=20] node[above] {$\frac12 a$} (q1);
\path[->] (q1) edge [loop,out=250,in=290,looseness=20] node[below] {$\frac14 b$} (q1);
\path[->] (q2) edge [loop,out=110,in=70,looseness=20] node[above] {$\frac14 a$} (q2);
\path[->] (q2) edge [loop,out=250,in=290,looseness=20] node[below] {$\frac12 b$} (q2);
\path[->] (r1) edge [loop,out=110,in=70,looseness=20] node[above] {$1 c$} (r1);
\path[->] (r2) edge [loop,out=110,in=70,looseness=20] node[above] {$1 c$} (r2);
\end{tikzpicture}
}
\end{center}
\vspace{-2mm}
\caption{Two LMCs.}
\label{fig-intro-irrational}
\end{figure}
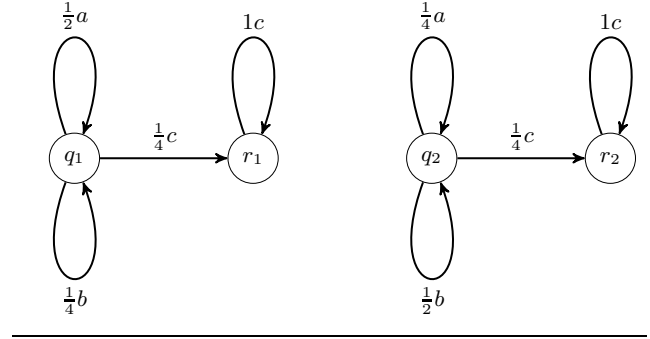
The semantics is as follows:
The chain starts in a given initial state (or in a random state according to a given initial distribution),
 picks a random transition according to the state's distribution over the outgoing transitions,
 outputs the letter of the transition,
 moves to the target state, and repeats.
In such a way, the chain produces a random infinite sequence of letters, i.e., a random infinite word.
We regard this infinite word as ``observable'' to the environment,
 whereas the infinite sequence of states remains ``internal'' to the chain.
Formally, an LMC defines a probability space whose samples are infinite words (also called \emph{runs} later) over~$\Sigma$. In \cite{GlabbeekSS95}, it is classified as a generative model.
LMCs appear as \emph{hidden Markov models} in speech recognition and in several areas of computational biology,
 cf.~\cite{LyngsoP02}.
LMCs, sometimes in the form of \emph{probabilistic automata} \cite{Rabin63},
 are also fundamental for modelling probabilistic systems.

Checking whether two LMCs (or, similarly, two probabilistic automata) are \emph{(language) equivalent}
 is a classical problem, going back to the seminal work of Sch\"{u}tzenberger~\cite{Schutzenberger} and Paz~\cite{Paz71}.
More recently, this problem was revisited, as various verification problems on probabilistic systems can be reduced to it (see, e.g., \cite{13KMOWW-LMCS}).
As a consequence, efficient polynomial-time algorithms and tools for equivalence checking
 have been developed \cite{CortesMRdistance,DoyenHenzingerRaskin,Cav11,13KMOWW-LMCS}.
If two systems are found to be \emph{not} equivalent, the question arises on \emph{how different} they are.
The \emph{distance} of two LMCs provides a measure for their difference,
with the extreme cases being distance~$0$ for equivalence and distance~$1$ for (almost-sure) distinguishability. 
The \emph{total variation distance}, which is a standard distance measure \cite{GibbsSuMetrics} between two probability distributions, yields a natural measure of the distance of two LMCs.
Given two probability distributions $\pi_1$ and~$\pi_2$ over the same \emph{countable} set $\Omega$, 
the total variation distance 
 is defined as
\begin{equation}
 d(\pi_1, \pi_2) := \max_{E \subseteq \Omega} |\pi_1(E)-\pi_2(E)| \,. \label{eq-intro-max}
\end{equation}
In words, $d(\pi_1, \pi_2)$ is the largest possible difference between probabilities that $\pi_1$ and~$\pi_2$ assign to the same event.
Furthermore, we have $d(\pi_1, \pi_2) = \pi_1(E) - \pi_2(E)$ for
\begin{equation}
 E = \{ r \in \Omega \mid \pi_1(r) \ge \pi_2(r) \} \;, \label{eq-intro-maximiser}
\end{equation}
so the event~$E$ is a maximizer in~\eqref{eq-intro-max}.
The total variation distance is---up to a factor of~$2$---equal
 to the \emph{$L_1$-norm} of the difference between $\pi_1$ and~$\pi_2$:
\[
 2 d(\pi_1, \pi_2) = \| \pi_1 - \pi_2 \|_1 := \sum_{x \in \Omega} | \pi_1(x) - \pi_2(x) |\,.
\]

When applying the total variance distance to LMCs, it should be emphasized that the sample space  $\Omega = \Sigma^\omega$
 (i.e.,  the set of infinite words over~$\Sigma$) is \emph{uncountable}.
Hence the maximum in the definition of (total variation) distance
 needs to be replaced by the supremum. Concretely, assume two LMCs $\M_1, \M_2$ with initial state distributions,
the LMCs assign each (measurable) event $E \subseteq \Sigma^\omega$ a probability $\pi_1(E)$ and $\pi_2(E)$, respectively.
So the (total variation) distance between $\M_1, \M_2$ is defined as
\begin{equation*}
 d(\pi_1, \pi_2) := \sup_{E \subseteq \Sigma^\omega} |\pi_1(E)-\pi_2(E)| \,.
\end{equation*}
It is not clear a priori if a maximizer event exists.
We will show later in this paper that it does exist.
In particular this means that $d(\pi_1, \pi_2) = 1$ holds if and only if there is an event~$E$ with $\pi_1(E) = 1$ and $\pi_2(E) = 0$.

While being an intriguing theoretical question, the study of the distance between LMCs also has practical implications.
For instance, 
in the verification of anonymity properties~\cite{Cav11,13KMOWW-LMCS}
 the following scenario is common: Two users are modelled as LMCs and leave a trace (i.e., emit a run).
An evil agent knows the two users, and sees a single trace. The agent wants to find out which of the two users has emitted the trace.
Clearly language equivalence (distance~$0$) of LMCs implies anonymity of the users.
If the distance is nonzero, one may ask if the agent can identify the users 
almost surely.
If the distance 
is~$1$, the agent succeeds with probability~$1$, 
because the agent can define an event~$E$ that occurs in the first LMC with probability~$1$, and in the second one with probability~$0$;
all the agent has to do is to check whether the given run belongs to~$E$. Conversely, if the distance is less than~$1$, the agent cannot almost-surely distinguish the users.
From this point of view, a distance less than~$1$ is a minimum requirement for some form of user anonymity, which could perhaps be called \emph{deniability}.

Another example is probabilistic model checking where 
computing the probability of certain events~$E$ is of central interest.
If the distance between some given LMCs is small (and known or bounded above),
 computing the probability of~$E$ in one of those chains may be enough
  for obtaining good bounds on the probability of~$E$ in the other chains.
This may lead to savings in the overall model-checking time.

\medskip

\noindent\textbf{Main Contributions.}
In this paper we develop a theory for the total variation distance between two LMCs. We pay special attention to the algorithmic and computational aspects of the problem. %
We make the following contributions:

\begin{enumerate}
\item[(1)] We demonstrate some basic properties of the total variation distance between two LMCs:
(a) the supremum in the definition can be ``achieved'', and we exhibit a maximizing event, although we show that
 the maximizing event is \emph{not} $\omega$-regular in general;
(b) the distance of two LMCs can be irrational even if all probabilities appearing in their description are rational.

\item[(2)] We study the \emph{qualitative} variant of the distance problem, i.e., to decide whether two LMCs have distance 1 or 0. The distance-0 problem amounts to the language equivalence problem for probabilistic automata, for which a polynomial-time algorithm exists. We provide a polynomial-time algorithm for the distance-$1$ problem.

\item[(3)] We study the \emph{quantitative} variant of the distance problem. In light of~(1), at best one can hope to \emph{approximate} the distance rather than to really \emph{compute} it (at least in the classical complexity theory framework). To this end, we provide an algorithm for approximating the distance with arbitrary precision. We also link the problem to \emph{Bernoulli convolutions} by providing an LMC where the distance of two states of this LMC is related to \emph{Bernoulli convolutions}, thus indicating the intricacy of the distance.

\item[(4)] We study the threshold problem, i.e., to decide whether the distance exceeds a given threshold. While leaving decidability of the problem open, we show that the problem is both NP-hard and hard for the square-root-sum problem.
\end{enumerate}

\noindent\textbf{Structure of the Paper.}
In Section~\ref{sec-prelim} we provide technical preliminaries.
In Section~\ref{sec-examples} we give two examples for LMCs and their distances.
In Section~\ref{sec-approx} we discuss two sequences that converge to the distance from below and from above,
 yielding an approximation algorithm.
In Section~\ref{sec-maximizing} we show that an event with maximum difference in probabilities always exists,
 and we exhibit such a ``witness'' event.
In Section~\ref{sec-irrational} we show that the distance can be irrational,
 and we give lower complexity bounds for the threshold problem. 
In particular, in Section~\ref{sec-function} we exhibit an LMC where the distance depends on the probabilities in the LMC
 in intricate ways, as witnessed by a connection to \emph{Bernoulli convolutions}. 
 In Section~\ref{sec-distance-1} we develop a polynomial-time algorithm for deciding whether two LMCs have distance~$1$.
In Section~\ref{sec-related-work} we discuss related work.
Finally, in Section~\ref{sec-conclusions} we offer some conclusions and highlight open problems.
Missing proofs can be found in \iftechrep{an appendix.}{in the full version of the paper~\cite{14CK-lics-report}.}

\section{Preliminaries} \label{sec-prelim}

We write $\N$ for the set of nonnegative integers.

Let $Q$ be a finite set.
By default we view \emph{vectors}, i.e., elements of $\R^Q$, as row vectors.
For a vector $\mu \in [0,1]^Q$ we write $|\mu| := \sum_{q \in Q} \mu(q)$ for its $L_1$-norm.
A vector $\mu \in [0,1]^Q$ is a \emph{distribution} (resp.\ \emph{subdistribution}) \emph{over~$Q$} if $|\mu| = 1$ (resp.\ $|\mu| \le 1$).
For $q \in Q$ we write $\delta_q$ for the (\emph{Dirac}) distribution over~$Q$ with $\delta_q(q) = 1$ and $\delta_q(r) = 0$ for $r \in Q \setminus\{q\}$.
For a subdistribution~$\mu$ we write $\supp(\mu) = \{q \in Q \mid \mu(q) > 0\}$ for its support.
Given two vectors $\mu_1, \mu_2 \in [0,1]^Q$ we write $\mu_1 \le \mu_2$ to say that $\mu_1(q) \le \mu_2(q)$ holds for all $q \in Q$.
We view elements of $\R^{Q \times Q}$ as \emph{matrices}.
A matrix $M \in [0,1]^{Q \times Q}$ is called \emph{stochastic} if each row sums up to one, i.e.,
 for all $q \in Q$ we have $\sum_{r \in Q} M(q,r) = 1$.

\begin{definition}
A \emph{labelled (discrete-time, finite-state) Markov chain (LMC)} is a tuple $\M = (Q, \Sigma, M)$ where
\begin{itemize}
 \item $Q$ is a finite set of states,
 \item $\Sigma$ is a finite alphabet of labels,
 and
 \item $M: \Sigma \to [0,1]^{Q \times Q}$ specifies the transitions,
  so that $\sum_{a \in \Sigma} M(a)$ is a stochastic matrix.
\end{itemize}
\end{definition}
Intuitively, if the LMC is in state~$q$, then with probability~$M(a)(q,q')$ it emits~$a$ and moves to state~$q'$.
For the complexity results of this paper, we assume that all the numbers in the matrices $M(a)$ for $a \in \Sigma$ are rationals given as fractions of integers represented in binary.
We extend $M$ to the mapping $M: \Sigma^* \to [0,1]^{Q \times Q}$ with $M(a_1 \cdots a_k) = M(a_1) \cdots M(a_k)$ for $a_1, \ldots, a_k \in \Sigma$.
Intuitively, if the LMC is in state~$q$ then with probability~$M(w)(q,q')$ it emits the word~$w$ and moves (in $|w|$ steps) to state~$q'$.

Fix an LMC~$\M = (Q, \Sigma, M)$ for the rest of this section.
A run of $\M$ is an infinite sequence $a_1 a_2 \cdots$ with $a_i \in \Sigma$ for all $i \in \N$. We write $\Sigma^\omega$ for the set of \emph{runs}.
For a run $r = a_1 a_2 \cdots$ and $i \in \N$ we write $r_i := a_1 a_2 \cdots a_i$.
For a set $W \subseteq \Sigma^*$ of finite words, we define $W \Sigma^\omega := \{w u \mid w \in W, \ u \in \Sigma^\omega\} \subseteq \Sigma^\omega$; i.e.,
 the set of runs that have a prefix in~$W$.
For $w \in \Sigma^*$ we define $\Run(w) := \{w\} \Sigma^\omega$;
 i.e, $\Run(w)$ is the set of runs starting with~$w$.

To an (initial) distribution $\pi$ over~$Q$ we associate the probability space $(\Sigma^\omega,\F,\Pr_\pi)$,
 where $\F$ is the $\sigma$-field generated by all \emph{basic cylinders} $\Run(w)$ with $w \in \Sigma^*$,
 and $\Pr_\pi: \F \to [0,1]$ is the unique probability measure such that
  $\Pr_\pi(\Run(w)) = |\pi M(w)|$.
We generalize the definition of~$\Pr_\pi$ to subdistributions~$\pi$ in the obvious way, yielding sub-probability measures.
An \emph{event} is a measurable set $E \subseteq \Sigma^\omega$.
In this paper we consider only measurable subsets of~$\Sigma^\omega$,
 and when we write $E \subseteq \Sigma^\omega$, the set~$E$ is meant to be measurable. An event is \emph{$\omega$-regular}, if it is equal to a language accepted by a nondeterministic B\"uchi automaton.
When confusion is unlikely, we may identify the \mbox{(sub-)}distribution~$\pi$ with the induced \mbox{(sub-)}probability measure~$\Pr_\pi$;
 i.e., for events $E \subseteq \Sigma^\omega$ we may write $\pi(E)$ for~$\Pr_\pi(E)$.
For a distribution~$\pi$ and a word $w \in \Sigma^*$, we write $\pi^w$ as a shorthand for $\pi M(w)$;
 intuitively this is the state subdistribution after emitting~$w$.
We have $\Pr_\pi(\Run(w)) = |\pi^w|$.

We reserve $\pi,\rho$ (and $\pi_1, \pi_2, \ldots$) for distributions over~$Q$,
 often viewing $\pi_1, \pi_2$ as given initial distributions.
Similarly, we reserve $\mu, \nu$ for subdistributions over~$Q$.
(But note that $\pi^w$ for $w \in \Sigma^*$ is a subdistribution in general.)

Given two initial distributions $\pi_1, \pi_2$,
 we 
 define the \emph{(total variation) distance} between $\pi_1$ and~$\pi_2$ by%
\[
 d(\pi_1, \pi_2) := \sup_{E \subseteq \Sigma^\omega} |\pi_1(E) - \pi_2(E)| \,.
\]
Recall that $E \subseteq \Sigma^\omega$ implicitly means that $E$ is measurable.
As $\pi_1(E) - \pi_2(E) = - (\pi_1(\Sigma^\omega \setminus E) - \pi_2(\Sigma^\omega \setminus E))$,
 we have in fact $d(\pi_1, \pi_2) = \sup_{E \subseteq \Sigma^\omega} (\pi_1(E) - \pi_2(E))$.

 \begin{remark} \label{rem:remark}
 One could analogously define the total variation distance between two LMCs $\M_1 = (Q_1, \Sigma, M_1)$ and $\M_2 = (Q_2, \Sigma, M_2)$
with initial distributions $\pi_1$ and~$\pi_2$ over $Q_1$ and~$Q_2$, respectively.
Our definition is without loss of generality, as one can take the LMC $\M = (Q, \Sigma, M)$ where $Q$ is the disjoint union of $Q_1$ and~$Q_2$,
and $M$ is defined using $M_1$ and~$M_2$ in the straightforward manner.
 \end{remark}

We write $\mu_1 \equiv \mu_2$ to denote that $\mu_1$ and~$\mu_2$ are \emph{(language) equivalent}, i.e.,
 that $|\mu_1^w| = |\mu_2^w|$ holds for all $w \in \Sigma^*$.
The following proposition states in particular that equivalence can be decided in polynomial time,
 and that equivalence and the distance being zero are equivalent.
\newcommand{\stmtpropequivalence}{
 \mbox{}
\begin{itemize}
 \item[(a)] We have $\pi_1 \equiv \pi_2$ if and only if $d(\pi_1, \pi_2) = 0$.
 \item[(b)] One can compute in polynomial time a set $\mathcal{B} \subseteq \Q^{2 |Q|}$ of column vectors, with $|\mathcal{B}| \le 2 |Q|$,
  such that for all subdistributions $\mu_1, \mu_2$ we have $\mu_1 \equiv \mu_2$
 if and only if $(\mu_1 \ \mu_2) \cdot b = 0$ holds for all $b \in \mathcal{B}$.
Here, $(\mu_1 \ \mu_2) \in [0,1]^{2 |Q|}$ is the row vector obtained by gluing $\mu_1, \mu_2$ together.
(Note that $(\mu_1 \ \mu_2) \cdot b$ is a scalar.)
 \item[(c)] We have $\mu_1 \equiv \mu_2$ if and only if $|\mu_1^w| = |\mu_2^w|$ holds for all $w \in \Sigma^*$ with $|w| = 2 |Q|$.
 \item[(d)] It is decidable in polynomial time whether $\mu_1 \equiv \mu_2$ holds.
       Hence it is also decidable in polynomial time whether $d(\pi_1, \pi_2) = 0$ holds.
\end{itemize}
}
\begin{proposition} \label{prop-equivalence}
\stmtpropequivalence
\end{proposition}
Proposition~\ref{prop-equivalence}~(a) is immediate from the definitions.
Parts (b)--(d) follow from a linear-algebra argument described, e.g., in \cite{Schutzenberger,Paz71,DoyenHenzingerRaskin}.
We sketch this argument \iftechrep{in Appendix~\ref{app-prelim}}{in~\cite{14CK-lics-report}}.

\section{Examples} \label{sec-examples}

We illustrate some phenomena of the distance by two examples. The main observations are that the distance of two LMCs can be irrational (Example~1), and in general,  they must be differentiated by events which are \emph{not} $\omega$-regular (Example~1), even if their distance is 1 (Example~2).

\subsection{Example 1} \label{example1}
Consider the LMCs from Figure~\ref{fig-intro-irrational} on page~\pageref{fig-intro-irrational}.
As discussed in Remark~\ref{rem:remark}, we can equivalently view them as a single LMC.
To illustrate the definitions we study the distance between states $q_1$ and~$q_2$,
 or more precisely, between the Dirac distributions $\delta_{q_1}$ and~$\delta_{q_2}$.
Note that we have $\delta_{r_1} \equiv \delta_{r_2}$, as both $r_1$ and~$r_2$ keep emitting the letter~$c$.
On the other hand we have $\delta_{q_1} \not\equiv \delta_{q_2}$ and so $d(\delta_{q_1}, \delta_{q_2}) > 0$.
With probability~$1$, one of the states $r_1, r_2$ will eventually be reached.
So events are characterized by the words over $a,b$ emitted before the infinite $c$-sequence.
More formally, for any event $E \subseteq \Sigma^\omega = \{a,b,c\}^\omega$
 one can define $W_E := \{w \in \{a,b\}^* \mid w c^\omega \in E \}$ so that we have
 \[
  \delta_{q_i}(E) = \sum_{w \in W_E} \delta_{q_i}(\{w\} \{c\}^\omega) \quad \text{for $i \in \{1,2\}$.}
 \]
It is easy to see that $\delta_{q_1}(\{a\} \{c\}^\omega) = \frac18$ and $\delta_{q_2}(\{a\} \{c\}^\omega) = \frac{1}{16}$.
Consider any event~$E$ with $W_E$ defined as above.
If $a \in W_E$, then $\delta_{q_2}(E) \ge \delta_{q_2}(\{a\} \{c\}^\omega) = \frac{1}{16}$.
If $a \not\in W_E$, then $\delta_{q_1}(E) \le 1 - \delta_{q_1}(\{a\} \{c\}^\omega) = \frac78$.
So for any~$E$ we have $\delta_{q_1}(E) - \delta_{q_2}(E) \le \max\{1 - \frac{1}{16}, \frac78 \} = \frac{15}{16}$.
By symmetry we also have $\delta_{q_2}(E) - \delta_{q_1}(E) \le \frac{15}{16}$.
As $E$ was arbitrary, we have thus shown $d(\delta_{q_1}, \delta_{q_2}) \le \frac{15}{16} < 1$.
We will show in Proposition~\ref{prop-irrational-1} that we have in fact $d(\delta_{q_1}, \delta_{q_2}) = \sqrt{2}/{4}$,
 so distances may be irrational.
The proof of Proposition~\ref{prop-irrational-1} shows that
 $d(\delta_{q_1}, \delta_{q_2}) = \delta_{q_1}(E) - \delta_{q_2}(E)$ holds for the event
\[
 E := \{w c c c \ldots \mid w \in \{a,b\}^*, \ \#_a(w) \ge \#_b(w) \}\;,
\]
where $\#_a(w)$ and~$\#_b(w)$ denote the number of occurrences of $a$  and $b$ in~$w$ respectively.
This may be intuitive as $q_1$ is more likely to emit $a$-letters than $b$-letters, whereas for~$q_2$ it is the opposite.
We remark that this event~$E$ is not \emph{$\omega$-regular}, i.e.,
 it cannot be recognized by a B\"uchi automaton. As a matter of fact, any $\omega$-regular event can only differentiate the two LMCs by a rational number, as the probability of any $\omega$-regular event must be rational.


\subsection{Example 2}
Consider the LMC in Figure~\ref{fig-surprising}.
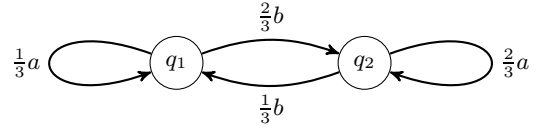
\begin{figure}
\begin{center}
\scalebox{1}{
\begin{tikzpicture}[scale=2.5,LMC style]
\node[state] (q1) at (0,-1) {$q_1$};
\node[state] (q2) at (1,-1) {$q_2$};
\path[->] (q1) edge [bend left=20] node[above] {$\frac23 b$} (q2);
\path[->] (q2) edge [bend left=20] node[below] {$\frac13 b$} (q1);
\path[->] (q1) edge [loop,out=160,in=200,looseness=20] node[left] {$\frac13 a$} (q1);
\path[->] (q2) edge [loop,out=20,in=-20,looseness=20] node[right] {$\frac23 a$} (q2);
\end{tikzpicture}
}
\end{center}
\vspace{-2mm}
\caption{A $2$-state LMC. The two states have distance~$1$.}
\label{fig-surprising}
\end{figure}
Both states $q_1, q_2$ can initiate \emph{any} run $r \in \Sigma^\omega$.
Note also that we have $\delta_{q_1}(\{r\}) = \delta_{q_2}(\{r\}) = 0$ for any single run $r \in \Sigma^\omega$.
Nevertheless it follows from Theorem~\ref{thm-limits-coincide} that we have $d(\delta_{q_1}, \delta_{q_2}) = 1$.
Moreover, Theorem~\ref{thm-maximising-event} will provide an event~$E$ with $\delta_{q_1}(E) = 1$ and $\delta_{q_2}(E) = 0$.
Intuitively, such an event could be based on the observation that if $q_1$ is the initial state,
 it is more likely after an even number of emitted $b$-letters to emit another~$b$,
whereas if $q_2$ is the initial state,
 it is more likely after an even number of emitted $b$-letters to emit an $a$-letter.
By the law of large numbers, this difference almost surely ``shows'' in the long run.

In the following we sketch a proof for the fact that
 no $\omega$-regular event~$E$ satisfies both $\delta_{q_1}(E) = 1$ and $\delta_{q_2}(E) = 0$.
In fact, we even show that for any $\omega$-regular~$E$ with $\delta_{q_1}(E) = 1$ we also have $\delta_{q_2}(E) = 1$.
(We omit precise automata-theoretic definitions here, as this argument will play no further role in this paper.)
Let $E$ be any $\omega$-regular event.
Let $R$ be a deterministic Rabin automaton for~$E$, with initial state~$r_0$.
Let $\M_R$ denote the LMC obtained by taking the cross-product of $R$ and the chain from Figure~\ref{fig-surprising}.
Let $\delta_{q_1}(E) = 1$.
Then all bottom SCCs of~$\M_R$ reachable from $(r_0,q_1)$ are accepting.
As the qualitative transition structure (i.e., distinguishing only zero and nonzero transition probabilities)
 is completely symmetric for $q_1$ and~$q_2$,
 it follows that all bottom SCCs of~$\M_R$ reachable from $(r_0,q_2)$ are accepting as well.
Hence we have $\delta_{q_2}(E) = 1$.

\section{An Approximation Algorithm} \label{sec-approx}

In this section we define two computable sequences that converge to the distance from below and from above, respectively.
This yields an algorithm for approximating the distance with arbitrary precision.

From now on until the end of Section~\ref{sec-maximizing} we fix an LMC $\M = (Q, \Sigma, M)$ and (initial) distributions $\pi_1,\pi_2$.
For $w \in \Sigma^*$ we define
\begin{align*}
 \mi(w)   & := \min\{|\pi_1^w|, |\pi_2^w|\} \quad \text{and} \quad  \\
 \con(w)   & := \max\{|\mu_1| \mid \mu_1 \le \pi_1^w \ \land \ \exists \mu_2 \le \pi_2^w : \mu_1 \equiv \mu_2\} \,.
\end{align*}
For $k \in \N$, we also define
$\mi(k)  := \sum_{w \in \Sigma^k} \mi(w)$ and
$\con(k) := \sum_{w \in \Sigma^k} \con(w) $.
The following proposition lists basic properties of those quantities.
\newcommand{\stmtpropbasic}{
Let $w \in \Sigma^*$ and $k \in \N$.
\begin{itemize}
 \item[(a)] We have $1 \ge \mi(w) \ge \con(w) \ge 0$.
  Hence $\mi(k) \ge \con(k)$.
 \item[(b)]
   We have $\mi(w) \ge \sum_{a \in \Sigma} \mi(w a) $ and
   $\con(w)  \le \sum_{a \in \Sigma} \con(w a)$.
  Hence we have
$\mi(k)   \ge \mi(k+1) $ and $\con(k)  \le \con(k+1)$
 \item[(c)]
  The limits
   $\mi(\infty)   := \lim_{i \to \infty} \mi(i)$ and $\con(\infty)  := \lim_{i \to \infty} \con(i)$
  exist,  and we have $\mi(\infty) \ge \con(\infty)$.
\end{itemize}
}
\begin{proposition} \label{prop-basic}
\stmtpropbasic
\end{proposition}
\begin{proof} \mbox{}
\begin{itemize}
\item[(a)]
We have $1 \ge |\pi_1^w| \ge \min\{|\pi_1^w|, |\pi_2^w|\} = \mi(w)$, hence $1 \ge \mi(w)$.
Clearly, $\con(w) \ge 0$.

Let $\mu_1, \mu_2$ be the subdistributions such that $\con(w) = |\mu_1|$ and $\pi_1^w \ge \mu_1 \equiv \mu_2 \le \pi_2^w$.
Then we have $|\pi_1^w| \ge |\mu_1| = |\mu_2| \le |\pi_2^w|$, hence $\con(w) = |\mu_1| \le \min\{|\pi_1^w|,|\pi_2^w|\} = \mi(w)$.
\item[(b)]
Let $i \in \{1,2\}$ with $|\pi_i^w| \le |\pi_{3-i}^w|$.
Then we have:
\begin{align*}
 & \mi(w) \\
 & = \min\{ |\pi_1^w|, |\pi_2^w| \} = |\pi_i^w| = |\pi_i M(w)| \\
 & = |\pi_i M(w) \sum_{a \in \Sigma} M(a)| \\ 
 & = \sum_{a \in \Sigma} |\pi_i M(w a)| = \sum_{a \in \Sigma} |\pi_i^{w a}| \\
 & \ge \sum_{a \in \Sigma} \min\{ |\pi_i^{w a}|, |\pi_{3-i}^{w a}| \}  =   \sum_{a \in \Sigma} \mi(w a)
\end{align*}
Let $\mu_1, \mu_2$ be the subdistributions such that $\con(w) = |\mu_1|$ and $\pi_1^w \ge \mu_1 \equiv \mu_2 \le \pi_2^w$.
It follows that, for all $a \in \Sigma$, we have $\pi_1^{w a} \ge \mu_1 M(a) \equiv \mu_2 M(a) \le \pi_2^{w a}$,
 hence $\con(w a) \ge |\mu_1 M(a)|$.
So we have:
\begin{align*}
 & \sum_{a \in \Sigma} \con(w a) \\
 & \ge \sum_{a \in \Sigma} |\mu_1 M(a)| = |\mu_1 \sum_{a \in \Sigma} M(a)| = |\mu_1| \\
 & = \con(w)\,.
\end{align*}
\item[(c)] Follows from (a)~and~(b).
\end{itemize}
\end{proof}
The quantities $\mi(k)$ and~$\con(k)$ provide lower and upper bounds for the distance:
\begin{proposition} \label{prop-bounds}
For all $k \in \N$ we have:
\[
 1 - \mi(k) \quad \le \quad d(\pi_1, \pi_2) \quad \le \quad 1 - \con(k)\,.
\]
\end{proposition}
\begin{proof}
We show first the lower bound.
Let $k \in \N$. 
Define $W_1 := \{w \in \Sigma^k \mid |\pi_1^w| \ge |\pi_2^w| \}$ and
       $W_2 := \{w \in \Sigma^k \mid |\pi_1^w|  <  |\pi_2^w| \}$.
By the definitions we have:
\begin{align*}
 d(\pi_1, \pi_2)
 &  =  \sup_{E \subseteq \Sigma^\omega} \left( \pi_1(E)-\pi_2(E) \right) \\
 & \ge \pi_1(W_1 \Sigma^\omega) - \pi_2(W_1 \Sigma^\omega) \\
 &  =  1 - \pi_1(W_2 \Sigma^\omega) - \pi_2(W_1 \Sigma^\omega) \\
 &  =  1 - \sum_{w \in W_2} |\pi_1^w| - \sum_{w \in W_1} |\pi_2^w| \\
 &  =  1 - \sum_{w \in \Sigma^k} \min\{ |\pi_1^w|, |\pi_2^w| \} \\
 &  =  1 - \mi(k)\,.
\end{align*}

Now we show the upper bound.
For an event $E \subseteq \Sigma^\omega$ and a word $w \in \Sigma^*$, we denote by $w^{-1} E$ the event $\{u \in \Sigma^\omega \mid w u\in E\}$.
For $w \in \Sigma^*$ we write $\mu_1^{(w)}$ and $\mu_2^{(w)}$ to denote subdistributions
 with $\con(w) = |\mu_1^{(w)}| = |\mu_2^{(w)}|$ and $\pi_1^w \ge \mu_1^{(w)} \equiv \mu_2^{(w)} \le \pi_2^w$.
The following inequalities hold:
\begin{equation} \label{eq-prop-bounds}
\begin{aligned}
 \pi_1^w(w^{-1} E) & = \mu_1^{(w)}(w^{-1} E) + (\pi_1^w - \mu_1^{(w)})(w^{-1} E) \\
                   & \le \mu_1^{(w)}(w^{-1} E) + |\pi_1^w| - |\mu_1^{(w)}| \\
 \pi_2^w(w^{-1} E) & \ge \mu_2^{(w)}(w^{-1} E)
\end{aligned}
\end{equation}
We have:
\begin{align*}
  & d(\pi_1, \pi_2) \\
  &  =  \sup_{E \subseteq \Sigma^\omega} \left( \pi_1(E)-\pi_2(E) \right) \\
  &  =  \sup_{E \subseteq \Sigma^\omega} \sum_{w \in \Sigma^k} \pi_1^w(w^{-1}E)-\pi_2^w(w^{-1}E) \\
  & \le \sup_{E \subseteq \Sigma^\omega} \sum_{w \in \Sigma^k} \mu_1^{(w)}(w^{-1}E) + |\pi_1^w|-|\mu_1^{(w)}| \\[-2mm]
  & \hspace{20mm} \mbox{} - \mu_2^{(w)}(w^{-1}E)
     && \text{\hspace{-10mm}(by~\eqref{eq-prop-bounds})} \\[2mm]
  &  =  \sup_{E \subseteq \Sigma^\omega} \sum_{w \in \Sigma^k} |\pi_1^w|-|\mu_1^{(w)}|
     && \text{\hspace{-10mm}(as $\mu_1^{(w)} \equiv \mu_2^{(w)}$)} \\
  &  =  1 - \sum_{w \in \Sigma^k} |\mu_1^{(w)}| =  1-\con(k)
\end{align*}
\end{proof}
The lower bound in this proposition follows by considering the event $E_k := W_1 \Sigma^\omega$ (where $W_1$ is from the proof),
 which depends only on the length-$k$ prefix of the run.
In fact, if we restrict each run to its length-$k$ prefix,
 we obtain a finite sample space,
 and the event~$E_k$ is the maximizer according to~\eqref{eq-intro-maximiser} in the introduction.
We could define, for each $k \in \N$, a distance $d_k(\pi_1, \pi_2)$ with
\begin{align*}
 d_k(\pi_1, \pi_2) & = \max_{W \in \Sigma^k} |\pi_1(W \Sigma^\omega) - \pi_2(W \Sigma^\omega)| \\
                   & = \pi_1(E_k) - \pi_2(E_k) = 1 - \mi(k) \,.
\end{align*}
Since $d_k(\pi_1, \pi_2) \le d_{k+1}(\pi_1, \pi_2)$ holds by Proposition~\ref{prop-basic}~(b),
 there is a limit $\lim_{k \to \infty} d_k(\pi_1, \pi_2)$, which equals $d(\pi_1, \pi_2)$
  (as we will show in Theorem~\ref{thm-limits-coincide}).
This would offer an alternative but equivalent definition of the distance,
 which avoids the use of infinite runs by replacing them with increasing prefixes.

By combining Propositions \ref{prop-basic}~and~\ref{prop-bounds} we obtain
\begin{equation}
 1 - \mi(\infty) \quad \le \quad d(\pi_1, \pi_2) \quad \le \quad 1 - \con(\infty)\,. \label{eq-lower-upper-infinity}
\end{equation}
In the rest of this section we show that those inequalities are in fact equalities.

Recall that for a (random) run $r \in \Sigma^\omega$ we write $r_i \in \Sigma^i$ for the length-$i$ prefix of~$r$.
For $i \in \N$, we define the random variable~$L_i$ that assigns to a run $r \in \Sigma^\omega$ the likelihood ratio
\[
 L_i(r) := |\pi_2^{r_i}| / |\pi_1^{r_i}|\,.
\]
Observe that $L_0(r) = 1$.
\begin{proposition} \label{prop-lim-exists}
 We have
 \begin{align*}
   & \pi_1\left(\lim_{i \to \infty} L_i \text{ exists and is in } [0,\infty)\right) = 1 \quad \text{and} \\
   & \pi_2\left(\lim_{i \to \infty} 1/L_i \text{ exists and is in } [0,\infty)\right) = 1\,.
 \end{align*}
\end{proposition}
\begin{proof}
\newcommand{\Ex}{\textup{Ex}}
We prove only the first equality; the second equality is proved similarly.
First we show that the sequence $L_0, L_1, \ldots$ is a martingale.
Denote by $\Ex_1$ the expectation with respect to~$\pi_1$.
Let $i \in \N$ and let $w \in \Sigma^i$ with $|\pi_1^w| > 0$.
We have:
\begin{align*}
 & \Ex_1( L_{i+1} \mid \Run(w) ) \\
 & = \sum_{q, q' \in Q} \sum_{a \in \Sigma} \frac{\pi_1^w(q) M(a)(q,q')}{|\pi_1^w|} \cdot \frac{|\pi_2^{w a}|}{|\pi_1^{w a}|} \\
 & = \frac{1}{|\pi_1^w|} \sum_{a \in \Sigma} \frac{|\pi_2^{w a}|}{|\pi_1^{w a}|}
       \underbrace{\sum_{q \in Q} \pi_1^w(q) \sum_{q' \in Q} M(a)(q,q')}_{= |\pi_1^{w a}|} \\
 & = \frac{1}{|\pi_1^w|} \sum_{a \in \Sigma} |\pi_2^{w a}| = |\pi_2^{w}| / |\pi_1^w| = L_i
\end{align*}
So $L_0, L_1, \ldots$ is a martingale.
More precisely, the sequence $L_0, L_1, \ldots$ is a nonnegative martingale with $\Ex_1(L_i) = 1$ for all $i \in \N$.
So the martingale convergence theorem (more precisely, ``Doob's forward convergence theorem'', see e.g.~\cite{book:Williams})
 applies, and we obtain $\pi_1\left(\lim_{i \to \infty} L_i \text{ exists and is finite}\right) = 1$.
\end{proof}
In the following we may write $\lim_{i \to \infty} L_i = \infty$ to mean $\lim_{i \to \infty} 1/L_i = 0$.
Define
\begin{equation}
 \bar{L} := \lim_{i \to \infty} L_i \in [0, \infty] \quad \text{(if the limit exists).} \label{eq-bar-L}
\end{equation}
The random variable~$\bar{L}$ plays a crucial role in the next section and is also used in the proof of the following theorem.
\newcommand{\stmtthmlimitscoincide}{
 We have
 \[
  1 - \mi(\infty) \ = \ d(\pi_1, \pi_2) \ = \ 1 - \con(\infty)\,.
 \]
}
\begin{theorem} \label{thm-limits-coincide}
 \stmtthmlimitscoincide
\end{theorem}

\begin{proof}[Proof sketch]
The proof \iftechrep{(Appendix~\ref{app-approx})}{(see~\cite{14CK-lics-report})}
is somewhat technical and we only give a sketch here.
Considering~\eqref{eq-lower-upper-infinity} it suffices to show that $\mi(\infty) = \con(\infty)$.
By Proposition~\ref{prop-basic} we have $\mi(k) \ge \con(k)$ for all $k$,
 so loosely speaking we have to show that for ``large''\footnote{%
In the rest of this proof sketch we gloss over the precise meaning of ``small'', ``not much larger'', etc.,
 and omit the quotation marks.%
}~%
$k$, $\mi(k)$ is not much larger than~$\con(k)$.
We first show that this holds for individual runs started from~$\pi_1$; more precisely, we show for all $\gamma > 0$ that
 \begin{equation} \label{eq-thm-limits-coincide-sketch}
 \begin{aligned}
  & \pi_1(\bar{L} > 0) \\
  & = \pi_1\left(\bar{L} > 0 \ \land\ \exists i \in \N : \mi(r_i) \le (1 + \gamma) \con(r_i) \right)\,.
 \end{aligned}
 \end{equation}
In words: Conditioned under the event $\{\bar{L} > 0\}$ the probability that eventually $\mi(r_i) \le (1 + \gamma) \con(r_i)$ holds is~$1$.
To show~\eqref{eq-thm-limits-coincide-sketch} we first show that conditioned under $\{\bar{L} > 0\}$
 we have with probability~$1$ that
 the distance between the distributions
 $\pi_1^{r_i} / |\pi_1^{r_i}|$ and $\pi_2^{r_i} / |\pi_2^{r_i}|$ converges to~$0$.
Using the fact that the set of distributions is compact, one can then show~\eqref{eq-thm-limits-coincide-sketch}.

To show that for large~$k$, $\mi(k)$ is not much larger than~$\con(k)$,
 we consider a partition $\Sigma^k = W_1 \cup W_2 \cup W_3$.
The set~$W_1$ contains the words~$w$ with small $|\pi_2^w| / |\pi_1^w|$.
So $\sum_{w \in W_1} \mi(w) \le \sum_{w \in W_1} |\pi_2^w|$ is small.
The set~$W_2$ contains the words~$w$ with $\mi(w) > (1 + \gamma) \con(w)$ and large $|\pi_2^w| / |\pi_1^w|$.
Runs with prefixes in~$W_2$ and $\bar{L} = 0$ are unlikely, as $L_k = |\pi_2^w| / |\pi_1^w|$ is large
 and $k$ is large and the sequence $L_0, L_1, \ldots$ converges to~$\bar{L}$ by Proposition~\ref{prop-lim-exists}.
Runs with prefixes in~$W_2$ and $\bar{L} > 0$ are also unlikely because of~\eqref{eq-thm-limits-coincide-sketch}.
So $\sum_{w \in W_2} \mi(w) \le \sum_{w \in W_2} |\pi_1^w|$ is small.
Finally, the set~$W_3$ contains the words~$w$ with $\mi(w) \le (1 + \gamma) \con(w)$ and large $|\pi_2^w| / |\pi_1^w|$.
So $\sum_{w \in W_3} \mi(w)$ is (for small~$\gamma$) not much larger than $\sum_{w \in W_3} \con(w) \le \con(k)$.
By adding the mentioned inequalities we obtain that $\mi(k) = \sum_{w \in \Sigma^k}$ is not much larger than~$\con(k)$.
\end{proof}

\begin{corollary} \label{cor-maximising-event}
 There is an algorithm that, given $\varepsilon > 0$, computes $a \in \Q$ such that $d(\pi_1, \pi_2) \in [a, a+\varepsilon]$.
\end{corollary}
\begin{proof}
By Proposition~\ref{prop-bounds} and Theorem~\ref{thm-limits-coincide} the sequences $(1 - \mi(k))_{k \in \N}$
 and $(1 - \con(k))_{k \in \N}$ converge to $d(\pi_1, \pi_2)$ from below and above, respectively.
For each~$k$, the values $\mi(k)$ and~$\con(k)$ are computable.
\end{proof}

In terms of the complexity of approximating the distance we have the following result:

\begin{proposition} \label{prop:approx-np-hard} Approximating the distance up to any $\varepsilon$ whose size is polynomial in the given LMC
 is NP-hard with respect to Turing reductions.
\end{proposition}
\begin{proof}
In \cite[Section 6]{LyngsoP02} (see also \cite[Theorem 7]{CortesMRdistance}), a reduction is given from the \emph{clique decision problem} to show that computing the distance in LMCs is NP-hard.
In their reduction the distance is rational and of polynomial size in the input.
Using the continued-fraction method (see e.g.\ Section 2.4 of~\cite{EtessamiY10} for an explanation) it follows that
 a polynomial-time algorithm (if it exists) for approximating the distance can be used to 
  construct a polynomial-time algorithm for computing the distance exactly.
Hence the conclusion.
\end{proof}
This NP-hardness result also follows from the proof of \cite[Theorem 10]{CortesMRdistance}.

\section{A Maximizing Event} \label{sec-maximizing}

The proof of Theorem~\ref{thm-limits-coincide} does not yield an event~$E_1$ with
 $\pi_1(E_1) - \pi_2(E_1) = d(\pi_1, \pi_2) \stackrel{\text{def}}{=} \sup_{E \subseteq \Sigma^\omega} \left| \pi_1(E) - \pi_2(E) \right|$.
In fact, it is not clear a priori whether such an event exists.
In this section we exhibit such a ``witness''~$E_1$.
It follows that the supremum from the definition of distance is in fact a maximum.

For some intuition recall from~\eqref{eq-intro-maximiser} in the introduction that in the countable case
 the event $E_1 = \{ r \in \Omega \mid \pi_1(r) \ge \pi_2(r)\}$ is the desired maximizer.
In the case of LMCs this does not work, since each individual run may have probability~$0$
 (as, e.g., in Figure~\ref{fig-surprising}).
However, by rewriting the inequality $\pi_1(r) \ge \pi_2(r)$ as $\pi_2(r) / \pi_1(r) \le 1$,
 one is tempted to guess that $\pi_2(r) / \pi_1(r)$ can be replaced by $\bar{L}(r)$ as defined in~\eqref{eq-bar-L}.
In the rest of the section we show that this intuition is correct.
Define the events
\begin{align*}
 E_1  := \left\{\bar{L} \le 1\right\} \quad \text{and} \quad
 E_2  := \left\{\bar{L} > 1\right\}\,.
\end{align*}
By Proposition~\ref{prop-lim-exists} we have
\begin{equation}
 \pi_1(E_1) + \pi_1(E_2) = \pi_2(E_1) + \pi_2(E_2) = 1\,. \label{eq-prop-lim-exists}
\end{equation}
The following lemma will suffice for showing that $E_1$ is the desired maximizer.
\begin{lemma} \label{lem-first-step}
We have $\pi_1(E_2) + \pi_2(E_1) \le \mi(\infty)$.
\end{lemma}
\begin{proof}
Towards a contradiction, suppose that this does not hold.
Then there is $k' \in \N$ with $\pi_1(E_2) + \pi_2(E_1) > \mi(k')$;
 hence there is $\gamma > 0$ with
\begin{equation}
 \pi_1(E_2) + \pi_2(E_1) > \mi(k') + 4 \gamma\,. \label{eq-first-step}
\end{equation}
Choose $\varepsilon \in (0, \gamma]$ small enough so that
\[
 \pi_1(\bar{L} \in (1,1+2\varepsilon]) \le \gamma\,.
\]
Using Proposition~\ref{prop-lim-exists}, choose $k \ge k'$ large enough so that we have
\begin{equation}
\begin{aligned}
 \pi_1(L_k \le 1 + \varepsilon \ \land \ \bar{L} > 1 + 2\varepsilon) & \le \gamma \quad \text{and} \\
 \pi_2(L_k > 1 + \varepsilon \ \land \ \bar{L} \le 1)                & \le \gamma\,.
\end{aligned}
\label{eq-lem-conv-squeeze}
\end{equation}
Then we have:
\begin{align*}
 & \ \pi_1(E_2) \\
 & =   \pi_1(\bar{L} \in (1,1+2\varepsilon]) + \pi_1(\bar{L} > 1 + 2\varepsilon) && \text{(def.~of~$E_2$)} \\
 & \le \gamma  + \pi_1(\bar{L} > 1 + 2\varepsilon)                               &&  \text{(choice of~$\varepsilon$)} \\
 & = \gamma  + \pi_1(L_k > 1+\varepsilon\ \land \ \bar{L} > 1 + 2\varepsilon) \\
 & \quad \mbox{} + \pi_1(L_k \le 1+\varepsilon \ \land \ \bar{L} > 1 + 2\varepsilon) && \text{} \\
 & \le \gamma  + \pi_1(L_k > 1+\varepsilon) \\
 & \quad \mbox{} + \pi_1(L_k \le 1+\varepsilon \ \land \ \bar{L} > 1 + 2\varepsilon) && \text{} \\
 & \le 2\gamma  + \pi_1(L_k > 1+\varepsilon) && \text{(by~\eqref{eq-lem-conv-squeeze})}
\intertext{Similarly we have:}
 & \ \pi_2(E_1) \\
 & \le \pi_2(L_k > 1+\varepsilon \ \land \ \bar{L} \le 1) \\
 & \quad \mbox{} + \pi_2(L_k \le 1+\varepsilon) && \text{(def.~of~$E_1$)} \\
 & \le \gamma + \pi_2(L_k \le 1+\varepsilon)    && \text{(by~\eqref{eq-lem-conv-squeeze})}
\end{align*}
By adding those two inequalities we obtain
\begin{equation}
\begin{aligned}
 & \pi_1(E_2) + \pi_2(E_1) \\
 & \le 3\gamma + \pi_1(L_k > 1+\varepsilon) + \pi_2(L_k \le 1+\varepsilon)\,.
\end{aligned}
\label{eq-first-step2}
\end{equation}

Define the partition of~$\Sigma^k$ in $\Sigma^k = W_1 \cup W_2 \cup W_3$ with
\begin{align*}
 W_1 &:= \{w \in \Sigma^k \mid |\pi_2^w| / |\pi_1^w| \le 1\} \\
 W_2 &:= \{w \in \Sigma^k \mid 1 < |\pi_2^w| / |\pi_1^w| \le 1+\varepsilon\} \\
 W_3 &:= \{w \in \Sigma^k \mid 1+\varepsilon < |\pi_2^w| / |\pi_1^w|\}\,.
\end{align*}
Then we have
\begin{align*}
 \pi_2(L_k \le 1) & = \sum_{w \in W_1} |\pi_2^w| \\
 \pi_2(1 < L_k \le 1+\varepsilon) & = \sum_{w \in W_2} |\pi_2^w| \le \sum_{w \in W_2} (1+\varepsilon) |\pi_1^w| \\
 \pi_1(L_k > 1+\varepsilon) & = \sum_{w \in W_3} |\pi_1^w|\,.
\end{align*}
By adding those (in)equalities we obtain
\begin{align*}
 & \pi_1(L_k > 1+\varepsilon) + \pi_2(L_k \le 1+\varepsilon) \\
 & \le (1+\varepsilon) \mi(k) \le \mi(k) + \varepsilon \le \mi(k) + \gamma \,.
\end{align*}
Combining this with~\eqref{eq-first-step2} yields
\[
 \pi_1(E_2) + \pi_2(E_1) \le \mi(k) + 4\gamma \le \mi(k') + 4\gamma\,,
\]
thus contradicting~\eqref{eq-first-step} as desired.
\end{proof}

Now we can prove that $E_1$ is the desired maximizing event.
\begin{theorem} \label{thm-maximising-event}
We have
 \[
  d(\pi_1, \pi_2) = \pi_1(E_1) - \pi_2(E_1)\,.
 \]
\end{theorem}
\begin{proof}
We have:
\begin{align*}
  d(\pi_1,\pi_2)
  & \ge \pi_1(E_1) - \pi_2(E_1)     && \text{(definition of distance)} \\
  &  =  1 - \pi_1(E_2) - \pi_2(E_1) && \text{(by~\eqref{eq-prop-lim-exists})} \\
  & \ge 1 - \mi(\infty)             && \text{(Lemma~\ref{lem-first-step})} \\
  &  =  d(\pi_1,\pi_2)              && \text{(Theorem~\ref{thm-limits-coincide})}
\end{align*}
\end{proof}

\section{Irrational Distances and Lower Bounds} \label{sec-irrational}

The following proposition shows that the distance can be irrational even if all numbers in the description of the LMC are rational.
\begin{figure}
\begin{center}
\scalebox{1.0}{
\begin{tikzpicture}[scale=2.5,LMC style]
\node[state] (q1) at (-1,0) {$q_1$};
\node[state] (q2) at (+1,0) {$q_2$};
\node[state] (r)  at ( 0,0) {$r$};
\path[->] (q1) edge node[above] {$x \ c$} (r);
\path[->] (q2) edge node[above] {$x \ c$} (r);
\path[->] (q1) edge [loop,out=70,in=110,looseness=20] node[above] {$\frac12 a$} (q1);
\path[->] (q1) edge [loop,out=290,in=250,looseness=20] node[below] {$(\frac12 - x) b$} (q1);
\path[->] (q2) edge [loop,out=110,in=70,looseness=20] node[above] {$(\frac12 - x) a$} (q2);
\path[->] (q2) edge [loop,out=250,in=290,looseness=20] node[below] {$\frac12 b$} (q2);
\path[->] (r) edge [loop,out=110,in=70,looseness=20] node[above] {$1 c$} (r);
\end{tikzpicture}
}
\end{center}
\caption{In this LMC, $x \in (0,\frac12)$ is a parameter.
For $x = \frac14$ the Dirac distributions $\delta_{q_1}$ and $\delta_{q_2}$ have distance $\sqrt{2}/4 \not\in \Q$.}
\label{fig-irrational-1}
\end{figure}
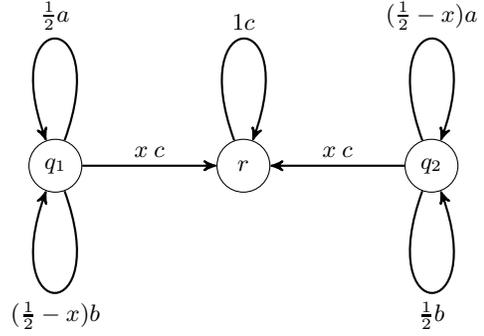
\newcommand{\stmtpropirrationalone}{
Consider the LMC shown in Figure~\ref{fig-irrational-1}, with parameter $x \in (0, \frac12)$.
We have $d(\delta_{q_1}, \delta_{q_2}) = \frac12 \sqrt{2 x}$.
}
\begin{proposition} \label{prop-irrational-1}
\stmtpropirrationalone
\end{proposition}

We start with a technical lemma.

\begin{lemma} \label{lem-irrational-binomial-series}
 For $y \in [0,\frac14)$ we have
 \[
  \sum_{n=0}^\infty {2 n \choose n} y^k \quad = \quad \frac{1}{\sqrt{1 - 4 y}}\,.
 \]
\end{lemma}
\begin{proof}
By a binomial series we have:
\begin{align*}
\frac{1}{\sqrt{1 - 4 y}}
  & = (1 - 4 y)^{-1/2}
   = \sum_{n=0}^\infty {-1/2 \choose n} (-4 y)^n
\end{align*}
By induction on $n \in \N$ one can show that ${-1/2 \choose n} = {2 n \choose n} \frac{(-1)^n}{4^n}$.
The lemma follows.
\end{proof}

\begin{proof}[Proof of Proposition~\ref{prop-irrational-1}]
We write $\pi_1 := \delta_{q_1}$ and $\pi_2 := \delta_{q_2}$.
Define $C := \{w c c c \ldots \mid w \in \{a,b\}^*\} \subseteq \Sigma^\omega$.
Clearly we have $\pi_1(C) = \pi_2(C) = 1$.
Define
 \begin{align*}
  E_> &:= \{w c c c \ldots \mid w \in \Sigma^*, \ \#_a(w) > \#_b(w) \} \quad \text{and} \\
  E_= &:= \{w c c c \ldots \mid w \in \Sigma^*, \ \#_a(w) = \#_b(w) \} \;,
 \end{align*}
 where $\#_a(w)$ and~$\#_b(w)$ denote the number of occurrences of $a$ resp.~$b$ in the word~$w$.
The events $E_\ge, E_<, E_\le$ are defined accordingly.

Recall the event $E_1 = \{\bar{L} \le 1\} \subseteq \Sigma^\omega$ from Section~\ref{sec-maximizing}.
Using the fact that the LMC in Figure~\ref{fig-irrational-1} is ``deterministic''
 (i.e., for each $a \in \Sigma$ and $q \in Q$ there is at most one $q' \in Q$ with $M(a)(q,q') > 0$),
 it is easy to verify that we have $E_1 \cap C = E_\ge$.
We have:
\begin{align}
 & \quad d(\pi_1, \pi_2) \notag \\
 & = \pi_1(E_1) - \pi_2(E_1)               && \text{(Theorem~\ref{thm-maximising-event})} \notag\\
 & = \pi_1(E_1 \cap C) - \pi_2(E_1 \cap C) && \text{(as $\pi_1(C) = \pi_2(C) = 1$)} \notag\\
 & = \pi_1(E_\ge) - \pi_2(E_\ge)           && \text{(as argued above)} \notag\\
 & = \pi_1(E_\ge) - \pi_1(E_\le)           && \text{(by symmetry of the chain)} \notag\\
 & = \pi_1(E_\ge) - (1 - \pi_1(E_>))       && \text{(as $\pi_1(C) = 1$)} \notag\\
 & = 2 \pi_1(E_>) + \pi_1(E_=) - 1         && \text{(by the definitions)}\,. \label{eq-irrational->=}
\end{align}
The following identity is proved in~\cite[p.167, (5.20)]{Concrete-Mathematics} and in~\cite{world-series} with a short combinatorial proof:
\begin{equation} \label{eq-knuth}
\sum_{m=0}^n {m+n \choose m} \left( \frac12 \right)^m \quad = \quad 2^n \qquad \text{for $n \in \N$.}
\end{equation}
For $m,n \in \N$ define $E(m,n) := \{w c c c \ldots \mid w \in \Sigma^*, \ \#_a(w) = m, \ \#_b(w) = n\}$.
We have:
\begin{align*}
 \pi_1(E_\le)
 & = \sum_{n=0}^\infty \sum_{m=0}^n \pi_1(E(m,n)) \\
 & = \sum_{n=0}^\infty \underbrace{\sum_{m=0}^n {m+n \choose m} \left( \frac12 \right)^m}_{= 2^n \text{ by~\eqref{eq-knuth}}} \left(\frac12 - x\right)^n x \\
 & = \frac{x}{1 - 2 (\frac12 -  x)} = \frac12
\end{align*}
So we have $\pi_1(E_>) = 1 - \pi_1(E_\le) = \frac12$ and hence by~\eqref{eq-irrational->=}
\begin{equation}
 d(\pi_1, \pi_2) = \pi_1(E_=) \,. \label{eq-irrational=}
\end{equation}
We have:
\begin{align*}
\pi_1(E_=)
& = \sum_{n=0}^\infty \pi_1(E(n,n)) \\
& = \sum_{n=0}^\infty {2 n \choose n} \left( \frac12 \right)^n \left(\frac12 - x\right)^n x \\
& = \frac{x}{\sqrt{1 - 4 \left( \frac14 - \frac12 x \right)}} = \frac12 \sqrt{2 x}  && \text{(by Lemma~\ref{lem-irrational-binomial-series})\,.}
\end{align*}
so the statement follows with~\eqref{eq-irrational=}.
\end{proof}

Note that when $x = \frac14$, the LMC shown in Figure~\ref{fig-irrational-1} is essentially
the union of the two LMCs shown in Figure~\ref{fig-intro-irrational}.
Proposition~\ref{prop-irrational-1} states that  $d(\delta_{q_1}, \delta_{q_2}) = \sqrt{2}/4$, thus substantiating a claim in Section~\ref{example1}.

This example suggests that in general it is not obvious what \emph{computing} the distance means, as it may be irrational.
Nevertheless it is shown in~\cite[Section~6]{LyngsoP02} that computing the distance is NP-hard (with respect to Turing reductions).
In that reduction the computed LMCs have a rational distance by construction.
However, in light of Proposition~\ref{prop-irrational-1} it may be more natural to study the \emph{threshold-distance} problem
 defined as follows:
Given an LMC, two initial distributions $\pi_1, \pi_2$, and a threshold $\tau \in [0,1] \cap \Q$,
 decide whether $d(\pi_1, \pi_2) \ge \tau$.

By Proposition~\ref{prop:approx-np-hard}, together with a binary search, the following lower bound follows:
\newcommand{\stmtthmNPhardness}{
The threshold-distance problem is NP-hard with respect to Turing reductions.
}
\begin{proposition} \label{prop-NP-hardness}
\stmtthmNPhardness
\end{proposition}
We remark that this can also be done by
modifying the reduction from~\cite{LyngsoP02}, \iftechrep{see Appendix~\ref{app-irrational}}{see~\cite{14CK-lics-report}}.

In the following we give another lower bound for the threshold-distance problem:
 the problem is hard for the square-root-sum problem, as we explain now.
Following~\cite{ABKM09} the \emph{square-root-sum} problem is defined as follows.
Given natural numbers $s_1, \ldots, s_n \in \N$ and $t \in \N$, decide whether $\sum_{i=1}^n \sqrt{s_i} \ge t$.
Membership of square-root-sum in NP has been open since 1976 when Garey, Graham and Johnson~\cite{GGJ76}
 showed NP-hardness of the travelling-salesman problem with Euclidean distances, but left membership in NP open.
It is known that square-root-sum reduces to PosSLP and hence lies in the 4th level of the counting hierarchy,
 see~\cite{ABKM09} and the references therein for more information on square-root-sum, PosSLP, and the counting hierarchy.

We use the LMC from Figure~\ref{fig-irrational-1} as a ``gadget'' to prove hardness for the square-root-sum problem:
\newcommand{\stmtthmsquarerootsum}{
There is a polynomial-time many-one reduction from the square-root-sum problem to the threshold-distance problem.
}
\begin{theorem} \label{thm-square-root-sum}
\stmtthmsquarerootsum
\end{theorem}
\begin{proof}[Proof sketch]
The construction is by taking the LMC from Figure~\ref{fig-irrational-1} as a gadget, and joining $n$ instances of it in parallel.
This is sketched for $n=3$ in Figure~\ref{fig-irrational-2}.
\begin{figure}
\begin{center}
\scalebox{0.78}{
\begin{tikzpicture}[scale=2.5, LMC style]
\node[state] (p1) at (-2,0) {$p_1$};
\node[state] (p2) at (+2,0) {$p_2$};
\node[state] (q11) at (-1,1.3) {$q_1^1$};
\node[state] (q12) at (-1,0) {$q_1^2$};
\node[state] (q13) at (-1,-1.3) {$q_1^3$};
\node[state] (q21) at (+1,1.3) {$q_2^1$};
\node[state] (q22) at (+1,0) {$q_2^2$};
\node[state] (q23) at (+1,-1.3) {$q_2^3$};
\node[state] (r)  at ( 0,0) {$r$};
\path[->] (p1) edge node[left] {$\frac13 c_1$} (q11.south west);
\path[->] (p1) edge node[above] {$\frac13 c_2$} (q12);
\path[->] (p1) edge node[left] {$\frac13 c_3$} (q13.north west);
\path[->] (p2) edge node[right] {$\frac13 c_1$} (q21.south east);
\path[->] (p2) edge node[above] {$\frac13 c_2$} (q22);
\path[->] (p2) edge node[right] {$\frac13 c_3$} (q23.north east);
\path[->] (q11) edge node[right] {$x_1 \ c$} (r.north west);
\path[->] (q12) edge node[above] {$x_2 \ c$} (r);
\path[->] (q13) edge node[right] {$x_3 \ c$} (r.south west);
\path[->] (q21) edge node[left] {$x_1 \ c$} (r.north east);
\path[->] (q22) edge node[above] {$x_2 \ c$} (r);
\path[->] (q23) edge node[left] {$x_3 \ c$} (r.south east);
\path[->] (q11) edge [loop,out=70,in=110,looseness=10] node[above] {$\frac12 a$} (q11);
\path[->] (q12) edge [loop,out=70,in=110,looseness=10] node[above] {$\frac12 a$} (q12);
\path[->] (q13) edge [loop,out=70,in=110,looseness=10] node[above] {$\frac12 a$} (q13);
\path[->] (q11) edge [loop,out=290,in=250,looseness=10] node[below] {$(\frac12 - x_1) b$} (q11);
\path[->] (q12) edge [loop,out=290,in=250,looseness=10] node[below] {$(\frac12 - x_2) b$} (q12);
\path[->] (q13) edge [loop,out=290,in=250,looseness=10] node[below] {$(\frac12 - x_3) b$} (q13);
\path[->] (q21) edge [loop,out=110,in=70,looseness=10] node[above] {$(\frac12 - x_1) a$} (q21);
\path[->] (q22) edge [loop,out=110,in=70,looseness=10] node[above] {$(\frac12 - x_2) a$} (q22);
\path[->] (q23) edge [loop,out=110,in=70,looseness=10] node[above] {$(\frac12 - x_3) a$} (q23);
\path[->] (q21) edge [loop,out=250,in=290,looseness=10] node[below] {$\frac12 b$} (q21);
\path[->] (q22) edge [loop,out=250,in=290,looseness=10] node[below] {$\frac12 b$} (q22);
\path[->] (q23) edge [loop,out=250,in=290,looseness=10] node[below] {$\frac12 b$} (q23);
\path[->] (r) edge [loop,out=110,in=70,looseness=10] node[above] {$1 c$} (r);
\end{tikzpicture}
}
\end{center}
\caption{This LMC is obtained by combining the chain from Figure~\ref{fig-irrational-1} in parallel $n=3$ times.
We have $d(\delta_{p_1}, \delta_{p_2}) =
  \frac13 \left( d(\delta_{q_1^1}, \delta_{q_2^1}) + d(\delta_{q_1^2}, \delta_{q_2^2}) + d(\delta_{q_1^3}, \delta_{q_2^3}) \right)$.}
\label{fig-irrational-2}
\end{figure}
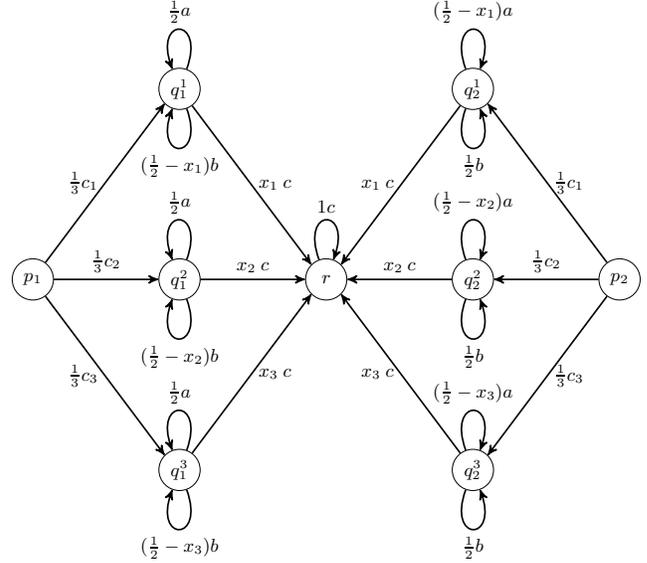
In general we have $\Sigma = \{c_1, \ldots, c_n, a, b, c\}$ and $Q = \{p_1, p_2, q_1^1, \ldots, q_1^n, q_2^1, \ldots, q_2^n, r\}$.
Using this construction we have
\begin{equation} \label{eq-square-root-1-mainbody}
d(\delta_{p_1}, \delta_{p_2}) = \frac{1}{n} \sum_{i=1}^n d(\delta_{q_1^i}, \delta_{q_2^i})\,.
\end{equation}
We prove~\eqref{eq-square-root-1-mainbody} \iftechrep{in Appendix~\ref{app-irrational}}{\cite{14CK-lics-report}}.
From the proof of Proposition~\ref{prop-irrational-1} we know the distances $d(\delta_{q_1^i}, \delta_{q_2^i})$
 and the corresponding maximizing events.
The proof is completed by suitably choosing the $x_i$ and the threshold~$\tau$,
 see \iftechrep{Appendix~\ref{app-irrational}}{\cite{14CK-lics-report}}.
\end{proof}

\subsection{Bernoulli Convolutions} \label{sec-function}
In this section, we establish another ``lower bound" by demonstrating a link to Bernoulli convolutions. 
Consider the LMC in Figure~\ref{fig-Bernoulli} which has two parameters: $\theta > 1$ and $x \in [-\frac12, \frac12]$.
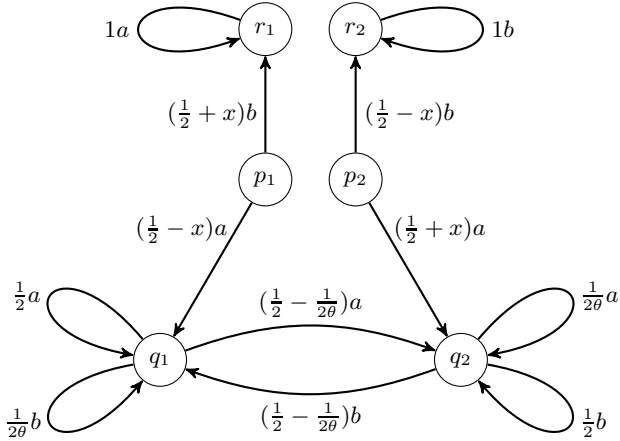
\begin{figure}
\begin{center}
\scalebox{1.0}{
\begin{tikzpicture}[scale=4,LMC style] 
\node[state] (r1) at (0.35,1.1) {$r_1$};
\node[state] (r2) at (0.65,1.1) {$r_2$};
\node[state] (p1) at (0.35,0.6) {$p_1$};
\node[state] (p2) at (0.65,0.6) {$p_2$};
\node[state] (q1) at (0,0) {$q_1$};
\node[state] (q2) at (1,0) {$q_2$};
\path[->] (q1) edge [bend left=20] node[above] {$(\frac12 - \frac1{2\theta}) a$} (q2);
\path[->] (q2) edge [bend left=20] node[below] {$(\frac12 - \frac1{2\theta}) b$} (q1);
\path[->] (q1) edge [loop,out=130,in=170,looseness=20] node[left] {$\frac12 a$} (q1);
\path[->] (q1) edge [loop,out=190,in=230,looseness=20] node[left] {$\frac1{2\theta} b$} (q1);
\path[->] (q2) edge [loop,out=50,in=10,looseness=20] node[right] {$\frac1{2\theta} a$} (q2);
\path[->] (q2) edge [loop,out=-10,in=-50,looseness=20] node[right] {$\frac12 b$} (q2);
\path[->] (p1) edge node[left,pos=0.2] {$(\frac12 - x) a$} (q1);
\path[->] (p2) edge node[right,pos=0.2] {$(\frac12 + x) a$} (q2);
\path[->] (p1) edge node[left,pos=0.4] {$(\frac12 + x) b$} (r1);
\path[->] (p2) edge node[right,pos=0.4] {$(\frac12 - x) b$} (r2);
\path[->] (r1) edge [loop,out=160,in=200,looseness=20] node[left] {$1 a$} (r1);
\path[->] (r2) edge [loop,out=20,in=340,looseness=20] node[right] {$1 b$} (r2);
\end{tikzpicture}
}
\end{center}
\caption{The distance between state $p_1, p_2$ depends on a \emph{Bernoulli-convolution}.}
\label{fig-Bernoulli}
\end{figure}
For each $\theta > 1$, denote by $d_\theta : [-\frac12, \frac12] \to [0,1]$ the function such that $d_\theta(x)$ is the distance between states $p_1$ and~$p_2$
 in the chain with parameters $\theta$ and~$x$.
Using the Banach fixed-point theorem one can show
\iftechrep{(see Appendix~\ref{app-function})}{(see~\cite{14CK-lics-report})}:
\newcommand{\stmtpropdistancefunction}{
For all $\theta > 1$ we have $d_\theta(x) = \frac12 + \frac12 f_\theta(x)$ for the unique function $f_\theta : \R \to \R$ with
\[
 f_\theta(x) = \begin{cases}
                 -2 x & x \le -\frac12 \\
                 \begin{array}{l}
                  \hspace{+1mm} \frac{1}{2 \theta} f_\theta( \theta x - (\frac12 \theta - \frac12)) \\
                  \hspace{-3mm} \mbox{} + \frac{1}{2 \theta} f_\theta( \theta x + (\frac12 \theta - \frac12))
                 \end{array} & x \in [-\frac12, \frac12] \\
                 2 x & x \ge +\frac12\;.
               \end{cases}
\]
}
\begin{proposition} \label{prop-distance-function}
\stmtpropdistancefunction
\end{proposition}
It follows that the derivative of~$f_\theta$ must satisfy
\begin{equation} \label{eq-f-prime-functional}
 f'_\theta(x) = \begin{cases}
                 -2 & x \le -\frac12 \\
                 \begin{array}{l}
                 \hspace{+1mm} \frac{1}{2} f'_\theta( \theta x - (\frac12 \theta - \frac12)) \\
                 \hspace{-3mm} \mbox{} + \frac{1}{2} f'_\theta( \theta x + (\frac12 \theta - \frac12))
                 \end{array} & x \in [-\frac12, \frac12] \\
                 2 & x \ge +\frac12\;.
               \end{cases}
\end{equation}
Again, one can use the Banach fixed-point theorem to show that the solution~$f'_\theta$ is unique.

The functional equation~\eqref{eq-f-prime-functional} is known from the study of \emph{Bernoulli convolutions},
 see~\cite{Bernoulli-60-years} for a survey and \cite[Chapter 5]{Experimental-Mathematics} for a gentle introduction.
In this field the solution of~\eqref{eq-f-prime-functional} occurs (translated and rescaled)
 as the cumulative distribution function of the random variable $\sum_{i=0}^\infty X_i/\theta^i$,
 where the $X_i$ are random variables that take on $-1$ and~$+1$ with probability~$\frac12$ each.
Bernoulli convolutions have been studied since the 1930s.
It is known that the solutions of~\eqref{eq-f-prime-functional} are either absolutely continuous or singular on~$[-\frac12, \frac12]$,
 depending on~$\theta$.
For $\theta > 2$ they are singular; in fact, for $\theta = 3$ the function is the (ternary) Cantor function.
For $\theta = 2$ we have $f'_2(x) = 4 x$ for $x \in [-\frac12, \frac12]$.
Erd\H{o}s showed that if $\theta$ is a \emph{Pisot number}\footnote{%
A Pisot number is a real algebraic integer greater than 1 such that all its Galois conjugates are less than 1 in absolute value.
The smallest Pisot number ($\approx 1.3247$) is the real root of $x^3 - x - 1$. Another one is the golden ratio $(\sqrt{5} + 1)/2 \approx 1.6180$.},
 then $f'_\theta$ is singular.
However, for almost all $\theta \in (1,2]$ the function~$f'_\theta$ is absolutely continuous.
It is open, e.g., for $\theta=3/2$ whether $f'_\theta$ is absolutely continuous or purely singular.

We conclude from this relation to Bernoulli convolutions that the distance can depend on the probabilities in the LMC in intricate ways.

\section{The Distance-1 Problem} \label{sec-distance-1}

The \emph{distance-1 problem} asks whether $d(\pi_1, \pi_2) = 1$
 holds for a given LMC and two distributions~$\pi_1, \pi_2$.
%
For the rest of the section
we fix an LMC~$\M = (Q, \Sigma, M)$ and initial distributions $\pi_1,\pi_2$.
Recall from Proposition~\ref{prop-equivalence} that $d(\pi_1, \pi_2) = 0$ is equivalent to $\pi_1 \equiv \pi_2$,
 and that the latter problem, language equivalence, is known to be decidable in polynomial time \cite{13KMOWW-LMCS}.
In this section we show that the distance-1 problem can also be decided in polynomial time.
The algorithm and its correctness argument are much more subtle.
The following proposition provides a characterisation of the 
case $d(\pi_1, \pi_2) < 1$.
\begin{proposition} \label{prop-dist-1-con}
We have $d(\pi_1, \pi_2) < 1$ if and only if
 there are $w \in \Sigma^*$ and subdistributions $\mu_1, \mu_2$ with $\mu_1 \le \pi_1^w$ and $\mu_2 \le \pi_2^w$ and $\mu_1 \equiv \mu_2$
  and $|\mu_1| = |\mu_2| > 0$.
\end{proposition}
Note that $\mu_1 \equiv \mu_2$ implies $|\mu_1| = |\mu_2|$.
Proposition~\ref{prop-dist-1-con} follows immediately from Theorem~\ref{thm-limits-coincide}.

Given $\pi_1, \pi_2$ and a word $w \in \Sigma^*$ one can compute $\pi_1^w$ and~$\pi_2^w$ in polynomial time.
Consider the following condition on~$w$:
\begin{equation} \label{eq-LP-1}
 \exists \mu_1, \mu_2 :
 \mu_1 \le \pi_1^w \text{ and } \mu_2 \le \pi_2^w \text{ and } \mu_1 \equiv \mu_2 \text{ and } |\mu_1| > 0\,.
\end{equation}
By Proposition~\ref{prop-equivalence}~(b), \eqref{eq-LP-1} amounts to a feasibility test of a linear program,
 and hence can be decided in polynomial time.
By Proposition~\ref{prop-dist-1-con} we have $d(\pi_1, \pi_2) < 1$ if and only if there is $w \in \Sigma^*$ such that \eqref{eq-LP-1} holds.

For notational convenience we write $\supp(w)$ for the pair $(\supp(\pi_1^w), \supp(\pi_2^w))$ in the following.
The condition~\eqref{eq-LP-1} on~$w$ is in fact only a condition on~$\supp(w)$,
 as $\mu_1 \equiv \mu_2$ implies $a \mu_1 \equiv a \mu_2$ for all $a \in [0, \infty)$.
So \eqref{eq-LP-1} can be rephrased as
\begin{equation}
\begin{aligned}
 \exists \mu_1, \mu_2 : \mbox{} & \supp(\mu_1) \subseteq  \supp(\pi_1^w) \text{ and } \supp(\mu_2) \subseteq \supp(\pi_2^w) \\
                                & \text{ and } \mu_1 \equiv \mu_2 \text{ and } |\mu_1| > 0\,.
\end{aligned}
\label{eq-LP-2}
\end{equation}
Moreover, for any two words $w, w' \in \Sigma^*$ with $\supp(w) = \supp(w')$
 we have $\supp(w a) = \supp(w' a)$ for all $a \in \Sigma$.
This implies
\[
\{ \supp(w) \mid w \in \Sigma^* \}  =  \{ \supp(w) \mid w \in \Sigma^*, \ |w| \le 2^{2 |Q|} \}\,.
\]
This suggests the following nondeterministic algorithm for checking whether $d(\pi_1, \pi_2) < 1$ holds:
 compute $\supp(w)$ for a guessed word $w \in \Sigma^*$ with $|w| \le 2^{2 |Q|}$ and check~\eqref{eq-LP-2} for feasibility.
Note that $w$ may have exponential length but need not be stored as a whole.
This results in a PSPACE algorithm.

In the following, we give a polynomial-time algorithm, which is based on further properties of the distance.

Given subdistributions $\mu_1, \mu_2$ with $|\mu_1|, |\mu_2| > 0$ we define the following relation:
\[
 \mu_1 \sim \mu_2 \quad \Longleftrightarrow \quad d \left( \frac{\mu_1}{|\mu_1|}, \frac{\mu_2}{|\mu_2|} \right) < 1
\]
Note that $\frac{\mu_1}{|\mu_1|}$ and~$\frac{\mu_2}{|\mu_2|}$ are distributions.
We have that $\mu_1 \equiv \mu_2$ implies $\mu_1 \sim \mu_2$.
The relation~$\mathord{\sim}$ is reflexive, symmetric, but in general not transitive.
We observe:
\begin{proposition} \label{prop-dist-1-wv}
Let $\mu_1 \sim \mu_2$.
Let $w \in \Sigma^*$ such that $\supp(\pi_1^w) \supseteq \supp(\mu_1)$ and $\supp(\pi_2^w) \supseteq \supp(\mu_2)$.
Then $d(\pi_1, \pi_2) < 1$.
\end{proposition}
\begin{proof}
Since $\mu_1 \sim \mu_2$,
we have $d(\rho_1, \rho_2)<1$ for the distributions $\rho_1 := \mu_1 / |\mu_1|$ and $\rho_2 := \mu_2 / |\mu_2|$.
By Proposition~\ref{prop-dist-1-con} there is a word $v \in \Sigma^*$ and subdistributions $\nu_1, \nu_2$
 with $|\nu_1| = |\nu_2| > 0$ and $\rho_1^v \ge \nu_1 \equiv \nu_2 \le \rho_2^v$.
Since $\supp(\pi_i^w) \supseteq \supp(\mu_i) = \supp(\rho_i)$ holds for $i \in \{1,2\}$, we get
 $\pi_1^{w v} \ge a \nu_1 \equiv a \nu_2 \le \pi_2^{w v}$ for some small enough $a > 0$.
Using Proposition~\ref{prop-dist-1-con} again it follows that $d(\pi_1, \pi_2) < 1$.
\end{proof}

The following proposition states two structural properties of the relation~$\mathord{\sim}$
 which can be proved using the fact that $d(\pi_1, \pi_2) = 1$ implies that there is a ``maximizing'' event~$E$ with
 $\pi_1(E) = 1$ and $\pi_2(E) = 0$, see Theorem~\ref{thm-maximising-event}.
\begin{proposition} \label{prop-dist-1-structure}
We have the following.
\begin{itemize}
\item[(a)]
Let $\mu_1 \equiv \mu_2$.
Let $\nu_1 \le \mu_1$ with $|\nu_1| > 0$.
Then $\nu_1 \sim \mu_2$.
\item[(b)]
Let $\mu_1 \sim \mu_2$.
Then there is $q \in \supp(\mu_1)$ with $\delta_q \sim \mu_2$.
\end{itemize}
\end{proposition}
\begin{proof}\mbox{}
\begin{itemize}
\item[(a)]
Towards a contradiction suppose that $d(\nu_1 / |\nu_1|, \mu_2 / |\mu_2|) = 1$.
Then by Theorem~\ref{thm-maximising-event} there is an event $E \subseteq \Sigma^\omega$ with
 $\frac{\nu_1}{|\nu_1|}(E) = 0$ and $\frac{\mu_2}{|\mu_2|}(E) = 1$, i.e., $\nu_1(E) = 0$ and $\mu_2(E) = |\mu_2|$.
We have:
\begin{align*}
 |\mu_2|
 & = \mu_2(E) \\
 & = \mu_1(E) && \text{(as $\mu_1 \equiv \mu_2$)} \\
 & = (\mu_1 - \nu_1)(E) + \nu_1(E) && \text{(as $\nu_1 \le \mu_1$)} \\
 & = (\mu_1 - \nu_1)(E) && \text{(as $\nu_1(E) = 0$)} \\
 & \le |\mu_1 - \nu_1| \\
 & = |\mu_1| - |\nu_1| && \text{(as $\nu_1 \le \mu_1$)} \\
 & < |\mu_1| && \text{(as $|\nu_1| > 0$)} \\
 & = |\mu_2| && \text{(as $\mu_1 \equiv \mu_2$)}\;,
\end{align*}
which is a contradiction. Hence $\nu_1 \sim \mu_2$.
\item[(b)]
Suppose that for all $q \in \supp(\mu_1)$ we have $\delta_q \not\sim \mu_2$, i.e., $d(\delta_q, \mu_2/|\mu_2|) = 1$.
By Theorem~\ref{thm-maximising-event} for all $q \in \supp(\mu_1)$ there is an event $E_q \subseteq \Sigma^\omega$ with
 $\delta_q(E_q) = 1$ and $\mu_2/|\mu_2|(E_q) = 0$.
Consider the event
 \[
  E := \bigcup_{q \in \supp(\mu_1)} E_q\,.
 \]
For all $q \in \supp(\mu_1)$ we have $\delta_q(E) \ge \delta_q(E_q) = 1$, so $\delta_q(E) = 1$.
Hence,
\[
 \mu_1(E) = \!\! \sum_{q \in \supp(\mu_1)} \! \mu_1(q) \delta_q(E) = \!\! \sum_{q \in \supp(\mu_1)} \! \mu_1(q) = |\mu_1|\,.
\]
On the other hand, by a union bound, we have
\[
 \mu_2(E) \le \sum_{q \in \supp(\mu_1)} \mu_2(E_q) = 0\,.
\]
If $|\mu_2| > 0$, then by the definition of the distance we have
 \[
  d\left(\frac{\mu_1}{|\mu_1|}, \frac{\mu_2}{|\mu_2|}\right) \ge \frac{\mu_1}{|\mu_1|}(E) - \frac{\mu_2}{|\mu_2|}(E) = 1-0 = 1\;,
 \]
 so $\mu_1 \not\sim \mu_2$.
If $|\mu_2| = 0$, then by the definition of~$\mathord{\sim}$ it also follows $\mu_1 \not\sim \mu_2$.
\end{itemize}
\end{proof}
For distributions $\pi_1, \pi_2$ we define a set $R^{\pi_1, \pi_2} \subseteq Q \times Q$:
\begin{align*}
R^{\pi_1, \pi_2} := \{(r_1,r_2) \in Q \times Q \mid  \exists w & \in \Sigma^* : r_1 \in \supp(\pi_1^w) \\
                                                               & \text{ and } r_2 \in \supp(\pi_2^w)\}
\end{align*}
This set can be computed in polynomial time:
\newcommand{\stmtlemdistoneR}{
Let $\pi_1, \pi_2$ be distributions.
Define a directed graph~$G$ as follows.
The vertex set is $Q \times Q$.
There is an edge from $(q_1, q_2) \in Q \times Q$ to $(r_1,r_2) \in Q \times Q$
 if there is $a \in \Sigma$ with $M(a)(q_1,r_1) > 0$ and $M(a)(q_2,r_2) > 0$.
Then we have:
\begin{align*}
 & R^{\pi_1, \pi_2} \\
 & = \{(r_1,r_2) \in Q \times Q \mid \mbox{} \exists q_1 \in \supp(\pi_1) \ \exists q_2 \in \supp(\pi_2):\\
 & \hspace{30mm}  (r_1,r_2) \text{ is reachable from $(q_1, q_2)$ in $G$}\}
\end{align*}
As a consequence, $R^{\pi_1, \pi_2}$ can be computed in polynomial time using graph reachability.
}
\begin{lemma} \label{lem-dist-1-R}
\stmtlemdistoneR
\end{lemma}
The proof of Lemma~\ref{lem-dist-1-R} is straightforward by induction. 
For $r_1 \in Q$ we define the projection $R^{\pi_1, \pi_2}_{r_1} := \{r_2 \in Q \mid (r_1, r_2) \in R^{\pi_1, \pi_2}\}$.
We are ready to show the main theorem of the section:
\begin{theorem} \label{thm-dist-1}
Let $\pi_1, \pi_2$ be distributions.
Then $d(\pi_1, \pi_2) < 1$ holds if and only if there are $r_1 \in Q$ and subdistributions $\mu_1, \mu_2$ such that
\[
 \mu_1 \equiv \mu_2 \text{ and } r_1 \in \supp(\mu_1) \text{ and } \supp(\mu_2) \subseteq R^{\pi_1, \pi_2}_{r_1} \,.
\]
\end{theorem}
\begin{proof}
($\Longrightarrow$) Let $d(\pi_1, \pi_2) < 1$.
Then by Proposition~\ref{prop-dist-1-con} there are $w \in \Sigma^*$ and subdistributions $\mu_1, \mu_2$ with
 $\mu_1 \le \pi_1^w$ and $\mu_2 \le \pi_2^w$ and $\mu_1 \equiv \mu_2$ and $|\mu_1| = |\mu_2| > 0$.
By the definition of~$R^{\pi_1, \pi_2}$ we have $(r_1, r_2) \in R^{\pi_1, \pi_2}$ for all $r_1 \in \supp(\mu_1)$ and all $r_2 \in \supp(\mu_2)$.
Choose any $r_1 \in \supp(\mu_1)$.
Then $\supp(\mu_2) \subseteq R^{\pi_1, \pi_2}_{r_1}$.

($\Longleftarrow$)
For the converse, let $r_1 \in Q$ and $\mu_1, \mu_2$ be subdistributions such that
\[
 \mu_1 \equiv \mu_2 \text{ and } r_1 \in \supp(\mu_1) \text{ and } \supp(\mu_2) \subseteq R^{\pi_1, \pi_2}_{r_1} \,.
\]
By Proposition~\ref{prop-dist-1-structure}~(a) we have $\delta_{r_1} \sim \mu_2$.
Hence, by Proposition~\ref{prop-dist-1-structure}~(b) there is $r_2 \in \supp(\mu_2)$ with $\delta_{r_1} \sim \delta_{r_2}$.
Since $(r_1, r_2) \in R^{\pi_1, \pi_2}$, we have by the definition of~$R^{\pi_1, \pi_2}$ that
 there is $w \in \Sigma^*$ with $r_1 \in \supp(\pi_1^w)$ and $r_2 \in \supp(\pi_2^w)$.
As $\delta_{r_1} \sim \delta_{r_2}$, Proposition~\ref{prop-dist-1-wv} implies $d(\pi_1, \pi_2) < 1$.
\end{proof}

We highlight the algorithmic nature of Theorem~\ref{thm-dist-1} in Algorithm~\ref{alg-dist-1}.
\begin{algorithm}
\textbf{procedure} distance1 \\
\textbf{input}: LMC $\M = (Q, \Sigma, M)$ \\
\mbox{}\hspace{9mm} initial distributions $\pi_1, \pi_2 \in [0,1]^Q$ \\
\textbf{output}: $d(\pi_1, \pi_2) = 1$ or $d(\pi_1, \pi_2) < 1$ \\[1mm]
compute $\mathcal{B} \subseteq \Q^{2 |Q|}$ from Proposition~\ref{prop-equivalence}~(b) \\
compute $R^{\pi_1, \pi_2}$ by graph reachability (Lemma~\ref{lem-dist-1-R}) \\
for $r_1 \in Q$ do \\
\ind if there exist subdistributions $\mu_1, \mu_2$ with \\
\indd $r_1 \in \supp(\mu_1)$  and  $\supp(\mu_2) \subseteq R^{\pi_1, \pi_2}_{r_1}$ \\
\inddd and $\forall\, b \in \mathcal{B}$ : $(\mu_1 \ \mu_2) \cdot b = 0$ \\
\ind (* this can be decided using linear programming *) \\
\ind then return ``$d(\pi_1, \pi_2) < 1$'' \\
\ind fi \\
od \\
return ``$d(\pi_1, \pi_2) = 1$''
\caption{Polynomial-time algorithm for deciding the distance-1 problem.}
\label{alg-dist-1}
\end{algorithm}

\section{Related Work} \label{sec-related-work}

Two LMCs have distance~$0$ if and only if they are language equivalent. We have discussed works on language equivalence in the introduction.
The papers \cite{LyngsoP02} and~\cite{CortesMRdistance} are closest to ours.
They investigate the $L_p$-distance between two hidden Markov models~\cite{LyngsoP02}
 and two probabilistic automata~\cite{CortesMRdistance}. Those models are similar to ours. The main difference is that in their models
 no letters are emitted once a special \emph{end state} is reached,  and the transition structure of the chains guarantees that an end state is eventually reached with probability~$1$.
 (Our model is more general, as one can make the LMC emit an
  infinite sequence of a special ``end letter'' once the end state is reached.)
So those models induce a probability distribution over~$\Sigma^*$, which makes the sample space countable.
As mentioned in the introduction, the $L_1$-distance is then twice the total variation distance,
 so the hardness results from~\cite{LyngsoP02,CortesMRdistance} become available:
  it is NP-hard to ``compute'' the $L_1$- and $L_\infty$-distance~\cite{LyngsoP02} and the $L_p$-distance for odd~$p$~\cite{CortesMRdistance},
   but recall our discussion after Proposition~\ref{prop-irrational-1} about irrational distances.
We note that the example from Figure~\ref{fig-irrational-1}, showing the existence of irrational distances,
 can be easily framed in their models. It is also shown in~\cite{CortesMRdistance} that it is NP-hard to approximate the $L_1$-distance within an additive error,
 and that the $L_p$-distance can be computed in polynomial time for even~$p$.

The total variation distance in LMCs is considered in~\cite{ChenBW12},
 where the authors give an upper bound on the total variation distance in terms of the \emph{bisimilarity pseudometric}
  defined in~\cite{DesharnaisGJPmetrics}.
Bisimilarity is a ``structural'' (i.e., based on the emitted letters and the states) notion of equivalence of LMCs,
 whereas language equivalence is purely ``semantical'' (based only on the emitted letters). Accordingly,
 the bisimilarity pseudometric defines a \emph{branching-time} distance while the total variance distance defines a \emph{linear-time} distance.
The authors of~\cite{ChenBW12} prove a quantitative analogue of the fact that bisimilarity implies language equivalence:
they prove that
 the bisimilarity pseudometric, which can be computed in polynomial time~\cite{ChenBW12},
  is an upper bound on the total variation distance which we discuss here.

\section{Conclusions and Open Problems} \label{sec-conclusions}

In this paper we have developed a theory of the total variation distance between two LMCs.
Two important theoretical results of this paper are summarized as:
\begin{enumerate}
\item[(1)]
 By considering longer and longer prefix words,
  one can define two sequences $(1 - \mi(i))_{i \in \N}$ and $(1 - \con(i))_{i \in \N}$ that
  converge to the distance from below and above, respectively.
\item[(2)]
 Using the martingale convergence theorem one can show that there is always a maximizing event,
  and we have explicitly exhibited one.
\end{enumerate}

These results have algorithmic consequences.
Our main algorithmic result is a procedure that decides the distance-1 problem in polynomial time.
The result~(1) 
also leads to an algorithm for approximating the distance with arbitrary precision.
%
We have also shown that the distance can be irrational, and we have given lower complexity bounds for the threshold-distance problem:
 it is NP-hard and hard for the square-root-sum problem.

The complexity and even the decidability of the threshold-distance problem are open problems.
A theoretical question is whether the distance is always algebraic.
We have established a connection to Bernoulli convolutions,
 the long history of which may hint at the difficulty of solving the mentioned open problems.

\subsection*{Acknowledgements.}
We would like to thank Christoph Haase for pointing us to~\cite{Concrete-Mathematics},
 James Worrell for valuable discussions, and anonymous referees for helpful comments.
Stefan Kiefer is supported by a Royal Society University Research Fellowship.

\bibliographystyle{abbrvnat}
\bibliography{db}

\iftechrep{%
\onecolumn
\newpage
\appendix
\section{Proof of Proposition~\ref{prop-equivalence}} \label{app-prelim}

\begin{qproposition}{\ref{prop-equivalence}}
\stmtpropequivalence
\end{qproposition}
\begin{proof} \mbox{}
\begin{itemize}
\item[(a)]
Immediate from the definitions.
\item[(b)]
Define
\[\eta := (\underbrace{1, \ldots, 1}_{|Q| \text{ times}}, \underbrace{-1, \ldots, -1}_{|Q| \text{ times}})^T \;,\]
where the superscript~$T$ denotes transpose.
According to the definitions, we have $\mu_1\equiv \mu_2$ if and only if we have
\[\left( \begin{matrix} \mu_1 & \mu_2 \end{matrix} \right)
\cdot
\left( \begin{matrix} M(w) & 0\\
0 & M(w) \end{matrix} \right)\cdot \eta = 0
\]
for all $w \in \Sigma^*$.
We write
\[
G := \left\{\left( \begin{matrix} M(w) & 0\\
0 & M(w) \end{matrix} \right) \cdot \eta \mathrel{} \middle| \mathrel{} w\in \Sigma^* \right\}\,.\]
Observe that $G$ is a (column) vector space.
As the vectors are $2 |Q|$-dimensional, a basis of~$G$ contains at most $2 |Q|$ vectors.
It is shown, e.g., in~\cite{DoyenHenzingerRaskin} that one can compute a basis $\mathcal{B}$ for~$G$ in $O(|Q|^3)$ time.
It follows that $\mu_1\equiv \mu_2$ holds if and only if we have $\left( \begin{matrix} \mu_1 & \mu_2 \end{matrix} \right) \cdot b =0$ for all $b\in \mathcal{B}$.
\item[(c)]
The direction ``$\Longrightarrow$'' is immediate from the definitions.
For the converse, let $\mu_1 \not\equiv \mu_2$.
By the linear-algebra argument from part~(b) there is a word $u \in \Sigma^*$ with $|u| \le 2 |Q|$ and $|\mu_1^u| \ne |\mu_2^u|$.
Towards a contradiction suppose that for all $v \in \Sigma^*$ with $|v| = 2 |Q| - |u|$ we have $|\mu_1^{u v}| = |\mu_2^{u v}|$.
Now we have:
\begin{align*}
 |\mu_1^u|
 & = \sum_{v \in \Sigma^{2 |Q| - |u|}} |\mu_1^{u v}| && \text{(as $\sum_{a \in \Sigma} M(a)$ is stochastic)} \\
 & = \sum_{v \in \Sigma^{2 |Q| - |u|}} |\mu_2^{u v}| && \text{(by assumption)} \\
 & = |\mu_2^u|                                       && \text{(as $\sum_{a \in \Sigma} M(a)$ is stochastic)}
\end{align*}
This is a contradiction.
So there is $v \in \Sigma^*$ with $|v| = 2 |Q| - |u|$ and $|\mu_1^{u v}| \ne |\mu_2^{u v}|$.
\item[(d)]
Immediate from part~(b).
\end{itemize}
\end{proof}

\section{Proof of Theorem~\ref{thm-limits-coincide}} \label{app-approx}

Recall that for a (random) run $r \in \Sigma^\omega$ we write $r_i \in \Sigma^i$ for the length-$i$ prefix of~$r$.
We first prove the following lemma:

\begin{lemma} \label{lem-decrease}
 For all $\varepsilon > 0$ we have
 \[
  \pi_1(\bar{L} > 0) \quad = \quad \pi_1\left(\bar{L} > 0 \ \land\ \exists i \in \N : \mi(r_i)  \le (1+\varepsilon) \con(r_i) \right)\,.
 \]
\end{lemma}
\begin{proof}
For distributions $\rho_1, \rho_2$ define $\widetilde{d}(\rho_1, \rho_2) := \max_{w \in \Sigma^{2 |Q|}} \left( |\rho_1^w| - |\rho_2^w| \right)=\max_{w \in \Sigma^{2 |Q|}} \left(| |\rho_1^w| - |\rho_2^w| \right|)$.
Clearly $\widetilde{d}(\rho_1, \rho_2) \ge 0$.

For a given run~$r$ we define $\rho_{1,i}$ and $\rho_{2,i}$ for all $i \in \N$: let $\rho_{1,i} := \pi_1^{r_i} / |\pi_1^{r_i}|$
 and $\rho_{2,i} := \pi_2^{r_i} / |\pi_2^{r_i}|$.
Intuitively, $\rho_{1,i}$ (resp.\ $\rho_{2,i}$) is the state distribution in the first (resp.\ second) LMC,
 conditioned under having emitted the run prefix~$r_i$.
Define $\widetilde{u}_i \in \Sigma^{2 |Q|}$ so that
 $\widetilde{d}(\rho_{1,i}, \rho_{2,i}) = |\rho_{1,i}^{\widetilde{u}_i}| - |\rho_{2,i}^{\widetilde{u}_i}|$.
For arbitrary $i \in \N$ and $\delta > 0$ define the event $E_{i,\delta} := \{ \widetilde{d}(\rho_{1,i}, \rho_{2,i}) \ge \delta \}$.
For any run  $r\in E_{i,\delta}$
we have
$|\rho_{1,i}^{\widetilde{u}_i}| - |\rho_{2,i}^{\widetilde{u}_i}| \ge \delta$ and hence
\begin{align}
 |\rho_{1,i}^{\widetilde{u}_i}| & \ge \delta \quad \text{and} \label{eq-decrease1a} \\
 \frac{|\rho_{2,i}^{\widetilde{u}_i}|}{|\rho_{1,i}^{\widetilde{u}_i}|} & \le 1 - \frac{\delta}{|\rho_{1,i}^{\widetilde{u}_i}|}\leq 1-\delta\,. \label{eq-decrease1b}
\end{align}
It follows that
\begin{align*}
 \delta
 & \le \pi_1(r_{i+{2 |Q|}} = r_i \widetilde{u}_i \mid E_{i,\delta}) && \text{by~\eqref{eq-decrease1a}} \\
 & \le \pi_1\left( \begin{array}{l}
     |\pi_2^{r_{i+{2 |Q|}}}| = |\pi_2^{r_i \widetilde{u}_i}| = |\pi_2^{r_{i}}| |\rho_{2,i}^{\widetilde{u}_i}| \quad \land\\[1mm]
     |\pi_1^{r_{i+{2 |Q|}}}| = |\pi_1^{r_i \widetilde{u}_i}| = |\pi_1^{r_{i}}| |\rho_{1,i}^{\widetilde{u}_i}|
     \end{array} \mathrel{}\middle|\mathrel{} E_{i,\delta} \right) \\
 & \le \pi_1\left(L_{i+{2 |Q|}} \le (1-\delta) L_i \mid E_{i,\delta}\right) && \text{by~\eqref{eq-decrease1b}}
\end{align*}
In words: for those runs in $E_{i,\delta}$, the probability of a decrease of $L_i$ in the next $2 |Q|$ steps by at least~$\delta L_i$
 is bounded below by~$\delta$.
It follows that if 
$\widetilde{d}(\rho_{1,i}, \rho_{2,i}) \ge \delta$ holds for infinitely many~$i$, a positive limit $\bar{L}$ exists with probability~$0$:
\[
 \pi_1(\bar{L} > 0 \quad \land \quad \widetilde{d}(\rho_{1,i}, \rho_{2,i}) \ge \delta \text{ holds for infinitely many~$i$}) = 0\,.
\]
Since $\delta > 0$ was chosen arbitrarily, we have:
\begin{equation} \label{eq-decrease2}
 \pi_1\left(\bar{L} > 0 \quad \land \quad \lim_{i \to \infty} \widetilde{d}(\rho_{1,i}, \rho_{2,i}) = 0\right) \quad = \quad \pi_1(\bar{L} > 0) \,.
\end{equation}

Consider a run with $\lim_{i \to \infty} \widetilde{d}(\rho_{1,i}, \rho_{2,i}) = 0$.
Since the pairs $(\rho_{1,i}, \rho_{2,i})$ are elements of the compact set $C = \{(\rho_1, \rho_2) \mid |\rho_1| = |\rho_2| = 1\}$, by the Bolzano-Weierstrass theorem,
 there is a subsequence $i(0) < i(1) < \ldots$ and distributions $\rho_{1,*}, \rho_{2,*}$ such that
\[
 \lim_{j \to \infty} \left( \rho_{1,i(j)}, \rho_{2,i(j)} \right) = (\rho_{1,*}, \rho_{2,*})\,.
\]
It follows that for all $\varepsilon > 0$ there is $i \in \N$ with
\[
  \rho_{1,*} \le (1+\varepsilon) \rho_{1,i} =  (1+\varepsilon)\pi_1^{r_i} / |\pi_1^{r_i}| \quad \text{and} \quad
  \rho_{2,*} \le (1+\varepsilon) \rho_{2,i} = (1+\varepsilon) \pi_2^{r_i} / |\pi_2^{r_i}|
\]
and hence
\[
  \mi(r_i) \rho_{1,*} \le (1+\varepsilon) \pi_1^{r_i} \quad \text{and} \quad
  \mi(r_i) \rho_{2,*} \le (1+\varepsilon) \pi_2^{r_i}\,.
\]
Since $\widetilde{d}$ is a continuous function on~$C$, we have $\widetilde{d}(\rho_{1,*}, \rho_{2,*}) = \widetilde{d}(\lim_{i\rightarrow\infty} \rho_{1,i}, \lim_{i\rightarrow \infty}\rho_{2,i}) =\lim_{i\rightarrow \infty} \widetilde{d}(\rho_{1,i}, \rho_{2,i}) = 0$.
Hence, by Proposition~\ref{prop-equivalence}~(c), we have $\rho_{1,*} \equiv \rho_{2,*}$,
 and so $  \mi(r_i) \rho_{1,*} \equiv  \mi(r_i) \rho_{2,*}$.
It follows
\begin{equation} \label{eq-decrease3}
  \mi(r_i)\le (1+\varepsilon)\con(r_i)\,.
\end{equation}

Using those considerations we obtain:
\begin{align*}
 \pi_1(\bar{L} > 0)
 & = \pi_1\left(\bar{L} > 0 \ \land \ \lim_{i \to \infty} \widetilde{d}(\rho_{1,i}, \rho_{2,i}) = 0\right)
   && \text{by \eqref{eq-decrease2}} \\
 & \le \pi_1\left( \bar{L} > 0 \ \land \ \exists i \in \N : \mi(r_i)\le (1+\varepsilon)\con(r_i) \right)
   && \text{by \eqref{eq-decrease3}} \\
 & \le \pi_1\left( \bar{L} > 0 \right)
\end{align*}
\end{proof}

Now we are ready to prove that the limits $\mi(\infty)$ and~$\con(\infty)$ coincide:
\begin{lemma} \label{lem-limits-coincide}
We have $\mi(\infty) = \con(\infty)$.
\end{lemma}
\begin{proof}
Considering Proposition~\ref{prop-basic}~(c) it suffices to show $\mi(\infty) \le \con(\infty)$.
Towards a contradiction, suppose this does not hold.
Then we have $\mi(\infty) > \con(\infty)$.
So there exists $\delta \in (0,1]$ so that for all $k' \in \N$ we have
\begin{equation}
\mi(k') > \con(k') + 4 \delta. \label{eq-lem-limits-coincide0}
\end{equation}

In the following definition we write $v \prec w$ to denote that $v \in \Sigma^*$ is a proper prefix of~$w \in \Sigma^*$.
Define
 \[
  H_\delta := \{w \in \Sigma^* \mid \mi(w) \le (1+\delta) \con(w) \ \land \ \forall v \prec w : \mi(v) > (1+\delta) \con(v) \}\,.
 \]
By \eqref{eq-lem-limits-coincide0}, $H_\delta\neq \emptyset$.
By Lemma~\ref{lem-decrease} there is $k_0 \in \N$ so that for all $k \ge k_0$ we have
\begin{equation} \label{eq-lem-limits-coincide1}
 \pi_1(\bar{L} > 0 \ \land \ \forall i \le k: r_i \not\in H_\delta) \quad \le \quad \delta\,.
\end{equation}
Using Proposition~\ref{prop-lim-exists}, there is $k_1 \in \N$ such that for all $k \ge k_1$ we have
\begin{equation} \label{eq-lem-limits-coincide2}
 \pi_1(\bar{L} = 0 \ \land \ L_k > \delta) \quad \le \quad \delta\,.
\end{equation}
Choose $k \ge \max\{k_0,k_1\}$, so \eqref{eq-lem-limits-coincide1} and~\eqref{eq-lem-limits-coincide2} hold.
Define the partition of~$\Sigma^k$ in $\Sigma^k = W_1 \cup W_2 \cup W_3$ with
\begin{align*}
 W_1 &:= \{w \in \Sigma^k \mid |\pi_2^w| / |\pi_1^w| \le \delta\} \\
 W_2 &:= \{w \in \Sigma^k \mid |\pi_2^w| / |\pi_1^w| > \delta \ \land \ \forall i \le k: r_i \not\in H_\delta\} \\
 W_3 &:= \{w \in \Sigma^k \mid |\pi_2^w| / |\pi_1^w| > \delta \ \land \ \exists i \le k: r_i \in H_\delta\}
\end{align*}
We have:
\begin{align*}
 \sum_{w \in W_1} \mi(w)
 & = \sum_{w \in W_1} |\pi_2^w| \le \sum_{w \in W_1} \delta |\pi_1^w| \le \delta \\
 \sum_{w \in W_2} \mi(w)
 & \le \sum_{w \in W_2} |\pi_1^w| = \pi_1(W_2 \Sigma^\omega) \\
 & = \pi_1(\bar{L} > 0 \ \land \ W_2 \Sigma^\omega) \\ & \ + \pi_1(\bar{L} = 0 \ \land \ W_2 \Sigma^\omega) \\
 & \le \delta + \delta = 2\delta && \text{(by~\eqref{eq-lem-limits-coincide1} and~\eqref{eq-lem-limits-coincide2})} \\
 \sum_{w \in W_3} \mi(w)
 & \le \sum_{w \in H_\delta \text{ s.t.\ } |w| \le k} \mi(w) && \text{(Prop.~\ref{prop-basic}~(b))} \\
 & \le \sum_{w \in H_\delta \text{ s.t.\ } |w| \le k} (1+\delta) \con(w) && \text{(def.~of~$H_\delta$)} \\
 & \le (1+\delta) \con(k) \le \con(k) + \delta && \text{(Prop.~\ref{prop-basic}~(b))}
\end{align*}
Adding those inequalities yields
\[
 \mi(k) = \sum_{w \in W_1 \cup W_2 \cup W_3} \mi(w) \le \con(k) + 4\delta\,,
\]
thus contradicting~\eqref{eq-lem-limits-coincide0}, as desired.
\end{proof}

Now Theorem~\ref{thm-limits-coincide} from the main body of the paper follows:
\begin{qtheorem}{\ref{thm-limits-coincide}}
 \stmtthmlimitscoincide
\end{qtheorem}
\begin{proof}
Immediate by combining \eqref{eq-lower-upper-infinity} and Lemma~\ref{lem-limits-coincide}.
\end{proof}

\section{Proofs of Section~\ref{sec-irrational}} \label{app-irrational}

\subsection{Proof of Proposition~\ref{prop-NP-hardness}}

We prove Proposition~\ref{prop-NP-hardness} from the main body of the paper:

\begin{qproposition}{\ref{prop-NP-hardness}}
\stmtthmNPhardness
\end{qproposition}
\begin{proof}
We modify the proof of~\cite[Section~6]{LyngsoP02}.
Let us first sketch some important features of that proof.
It is a reduction from the \emph{clique} decision problem:
Given a graph~$G = (V,E)$ and a threshold $t \in \N$, decide whether $G$ has a clique of size at least~$t$.
The clique decision problem is known to be NP-complete.

The authors of~\cite{LyngsoP02} describe an LMC $\M_G$, computed from~$G$, such that $\M_G$ emits substrings of the word
 $a_1 a_2 \cdots a_{|V|}$, where $a_1, \ldots, a_{|V|} \in \Sigma$.
(In their model, an LMC stops emitting letters once a special end state is reached.
This can be simulated in our model by emitting an infinite sequence of a special ``end'' letter once the end state is reached.
In addition, rather than ``omitting'' letters from $a_1 a_2 \cdots a_{|V|}$ by means of $\varepsilon$-transitions,
 in our model we output a special ``blank'' symbol.
Those changes do not cause problems.)
Without loss of generality they assume that $V = \{1, 2, \ldots, |V|\}$.
\newcommand{\degr}{\mathit{deg}}%
Define $\gamma := \sum_{v \in V} 2^{\degr(v)}$, where $\degr(v)$ is the number of neighbours of~$v$ in~$G$.
The gadget~$\M_G$ has the following properties for all $\{i_1, \ldots, i_k\} \subseteq \{1, \ldots, |V|\}$:
\begin{itemize}
\item
 If $i_1 < \ldots < i_k$ and $\{i_1, \ldots, i_k\}$ is a clique in~$G$,
  then $a_{i_1} \cdots a_{i_k}$ is emitted with probability~$k/\gamma$;
\item
 otherwise $a_{i_1} \cdots a_{i_k}$ is emitted with probability~$0$.
\end{itemize}
For $j \in \{0, \ldots, |V|\}$, let $n_j$ denote the number of cliques of size~$j$.
Note that computing the maximal clique size amounts to computing the maximal~$j$ such that $n_j \ne 0$.

Based on the gadget~$\M_G$, the authors of~\cite{LyngsoP02} construct $|V|+1$ pairs of LMCs:
 $(\M_1^0, \M_2^0), \ldots, (\M_1^{|V|}, \M_2^{|V|})$.
For $i \in \{0, \ldots, |V|\}$, let $d_i$ denote the distance between $\M_1^i$ and~$\M_2^i$.
Further, for $i \le \gamma / 2^{|V|}$ define $b_i := 1 - \frac{i 2^{|V|}}{\gamma}$ and $c_i := 1/\gamma$;
for $i > \gamma / 2^{|V|}$ define $b_i := 1 - \frac{\gamma}{i 2^{|V|}}$ and $c_i := \frac{1}{i 2^{|V|}}$.
The pairs $(\M_1^i, \M_2^i)$ are constructed such that
\begin{equation}
 d_i = b_i + c_i \sum_{j=0}^{|V|} n_j |i-j| \label{eq-NP-hardness}
\end{equation}
holds for all $i \in \{0, \ldots, |V|\}$.
It is argued in~\cite{LyngsoP02} that once $d_0, \ldots, d_{|V|}$ are known, one can compute $n_0, \ldots, n_{|V|}$ by solving the
 linear equation system~\eqref{eq-NP-hardness}.
The largest clique size is then the maximal~$j$ such that $n_j \ne 0$.
This establishes a Turing reduction from the clique decision problem to computing the distance.
(Note that by~\eqref{eq-NP-hardness} the distances are rational in this reduction.)

For a Turing reduction from the clique decision problem to the threshold-distance problem,
 it now suffices to argue that one can compute all~$d_i$ with a polynomial number of threshold queries (``is the distance at least~$\tau$?'').
Indeed, consider again~\eqref{eq-NP-hardness}.
The term $\sum_{j=0}^{|V|} n_j |i-j|$ is an integer between $0$ and $(|V| + 1) \cdot 2^{|V|} \cdot |V|$.
Since $b_i$ and $c_i$ are fixed and known for each~$i$, each $d_i$ takes one of at most $(|V| + 1) \cdot 2^{|V|} \cdot |V| + 1$ values.
Using binary search, $\log_2 \left( (|V| + 1) \cdot 2^{|V|} \cdot |V| + 1 \right)$ distance queries suffice to determine~$d_i$.
Doing this for each~$i$ results in a polynomial number of threshold queries.
\end{proof}

\subsection{Proof of Theorem~\ref{thm-square-root-sum}}

We prove Theorem~\ref{thm-square-root-sum} from the main body of the paper.
\begin{qtheorem}{\ref{thm-square-root-sum}}
There is a polynomial-time many-one reduction from the square-root-sum problem to the threshold-distance problem.
\end{qtheorem}
\begin{proof}
The construction is by taking the LMC from Figure~\ref{fig-irrational-1} as a gadget, and joining $n$ instances of it in parallel.
This is sketched for $n=3$ in Figure~\ref{fig-irrational-2-app}.
\begin{figure}
\begin{center}
\begin{tikzpicture}[scale=2.5, LMC style]
\node[state] (p1) at (-2,0) {$p_1$};
\node[state] (p2) at (+2,0) {$p_2$};
\node[state] (q11) at (-1,1.3) {$q_1^1$};
\node[state] (q12) at (-1,0) {$q_1^2$};
\node[state] (q13) at (-1,-1.3) {$q_1^3$};
\node[state] (q21) at (+1,1.3) {$q_2^1$};
\node[state] (q22) at (+1,0) {$q_2^2$};
\node[state] (q23) at (+1,-1.3) {$q_2^3$};
\node[state] (r)  at ( 0,0) {$r$};
\path[->] (p1) edge node[left] {$\frac13 c_1$} (q11.south west);
\path[->] (p1) edge node[above] {$\frac13 c_2$} (q12);
\path[->] (p1) edge node[left] {$\frac13 c_3$} (q13.north west);
\path[->] (p2) edge node[right] {$\frac13 c_1$} (q21.south east);
\path[->] (p2) edge node[above] {$\frac13 c_2$} (q22);
\path[->] (p2) edge node[right] {$\frac13 c_3$} (q23.north east);
\path[->] (q11) edge node[right] {$x_1 \ c$} (r.north west);
\path[->] (q12) edge node[above] {$x_2 \ c$} (r);
\path[->] (q13) edge node[right] {$x_3 \ c$} (r.south west);
\path[->] (q21) edge node[left] {$x_1 \ c$} (r.north east);
\path[->] (q22) edge node[above] {$x_2 \ c$} (r);
\path[->] (q23) edge node[left] {$x_3 \ c$} (r.south east);
\path[->] (q11) edge [loop,out=70,in=110,looseness=10] node[above] {$\frac12 a$} (q11);
\path[->] (q12) edge [loop,out=70,in=110,looseness=10] node[above] {$\frac12 a$} (q12);
\path[->] (q13) edge [loop,out=70,in=110,looseness=10] node[above] {$\frac12 a$} (q13);
\path[->] (q11) edge [loop,out=290,in=250,looseness=10] node[below] {$(\frac12 - x_1) b$} (q11);
\path[->] (q12) edge [loop,out=290,in=250,looseness=10] node[below] {$(\frac12 - x_2) b$} (q12);
\path[->] (q13) edge [loop,out=290,in=250,looseness=10] node[below] {$(\frac12 - x_3) b$} (q13);
\path[->] (q21) edge [loop,out=110,in=70,looseness=10] node[above] {$(\frac12 - x_1) a$} (q21);
\path[->] (q22) edge [loop,out=110,in=70,looseness=10] node[above] {$(\frac12 - x_2) a$} (q22);
\path[->] (q23) edge [loop,out=110,in=70,looseness=10] node[above] {$(\frac12 - x_3) a$} (q23);
\path[->] (q21) edge [loop,out=250,in=290,looseness=10] node[below] {$\frac12 b$} (q21);
\path[->] (q22) edge [loop,out=250,in=290,looseness=10] node[below] {$\frac12 b$} (q22);
\path[->] (q23) edge [loop,out=250,in=290,looseness=10] node[below] {$\frac12 b$} (q23);
\path[->] (r) edge [loop,out=110,in=70,looseness=10] node[above] {$1 c$} (r);
\end{tikzpicture}
\end{center}
\caption{This LMC is obtained by combining the chain from Figure~\ref{fig-irrational-1} in parallel $n=3$ times.
We have $d(\delta_{p_1}, \delta_{p_2}) =
  \frac13 \left( d(\delta_{q_1^1}, \delta_{q_2^1}) + d(\delta_{q_1^2}, \delta_{q_2^2}) + d(\delta_{q_1^3}, \delta_{q_2^3}) \right)$.}
\label{fig-irrational-2-app}
\end{figure}
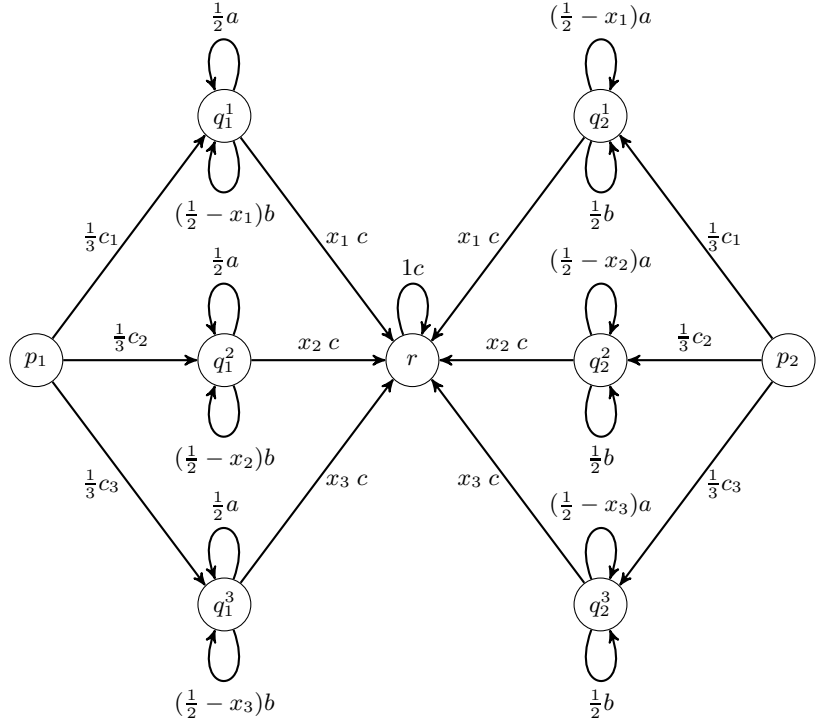
In general we have $\Sigma = \{c_1, \ldots, c_n, a, b, c\}$ and $Q = \{p_1, p_2, q_1^1, \ldots, q_1^n, q_2^1, \ldots, q_2^n, r\}$.
Using this construction we have
\begin{equation} \label{eq-square-root-1}
d(\delta_{p_1}, \delta_{p_2}) = \frac{1}{n} \sum_{i=1}^n d(\delta_{q_1^i}, \delta_{q_2^i})\,.
\end{equation}
To see this, consider the event
\[
  E_\ge := \{w c c c \ldots \mid w \in \{a,b\}^*, \ \#_a(w) \ge \#_b(w) \}
\]
from the proof of Proposition~\ref{prop-irrational-1}.
From that proof we know
\begin{equation} \label{eq-square-root-2}
 d(\delta_{q_1^i}, \delta_{q_2^i}) = \delta_{q_1^i}(E_\ge) - \delta_{q_2^i}(E_\ge)\,.
\end{equation}
It follows that we have
\begin{equation} \label{eq-square-root-3}
 d(\delta_{p_1}, \delta_{p_2}) = \delta_{p_1}(E) - \delta_{p_2}(E)
\end{equation}
for
\[
 E := \{\sigma w c c c \ldots \mid \sigma \in \{c_1, \ldots, c_n\}, \ w \in \{a,b\}^*, \ \#_a(w) \ge \#_b(w) \}\,.
\]
From the definition of~$E$ we have
\begin{equation} \label{eq-square-root-4}
 \delta_{p_1}(E) = \frac{1}{n} \sum_{i=1}^n \delta_{q_1^i}(E_\ge) \quad \text{and} \quad \delta_{p_2}(E) = \frac{1}{n} \sum_{i=1}^n \delta_{q_2^i}(E_\ge) \,.
\end{equation}
It follows:
\begin{align*}
d(\delta_{p_1}, \delta_{p_2})
& = \delta_{p_1}(E) - \delta_{p_2}(E) && \text{by~\eqref{eq-square-root-3}} \\
& = \frac{1}{n} \sum_{i=1}^n \left( \delta_{q_1^i}(E_\ge) - \delta_{q_2^i}(E_\ge) \right) && \text{by~\eqref{eq-square-root-4}} \\
& = \frac{1}{n} \sum_{i=1}^n d(\delta_{q_1^i}, \delta_{q_2^i})  && \text{by~\eqref{eq-square-root-2}\,,}
\end{align*}
hence \eqref{eq-square-root-1} is proved.

Recall that the input of the square-root-sum problem is a list of integers $s_1, \ldots, s_n \in \N$ and $t \in \N$.
Without loss of generality, we can assume that $s_1, \ldots, s_n, t \ge 1$.
The reduction is as follows.
Define $h := 3 \max_{i \in \{1, \ldots, n\}} s_i$.
Construct the LMC from Figure~\ref{fig-irrational-2-app} with $x_i := 2 s_i / h^2$.
Then we have $x_i \in (0, 1/2)$ and
\begin{equation}
 \frac12 \sqrt{2 x_i} = \frac{1}{h} \sqrt{s_i} \;. \label{eq-sqrt-reduction-1}
\end{equation}
Set the threshold $\tau := \frac{1}{n h} \cdot t$.
We have:
\begin{align*}
 d(\delta_{p_1}, \delta_{p_2})
 & = \frac{1}{n} \sum_{i=1}^n d(\delta_{q_1^i}, \delta_{q_2^i}) && \text{(by~\eqref{eq-square-root-1})} \\
 & = \frac{1}{n} \sum_{i=1}^n \frac12 \sqrt{2 x_i} && \text{(Proposition~\ref{prop-irrational-1})} \\
 & = \frac{1}{n h} \sum_{i=1}^n \sqrt{s_i} && \text{(by~\eqref{eq-sqrt-reduction-1})\,.}
\end{align*}
It follows that we have $d(\delta_{p_1}, \delta_{p_2}) \ge \tau$ if and only if $\sum_{i=1}^n \sqrt{s_i} \ge t$.
\end{proof}

\section{Proof of Proposition~\ref{prop-distance-function}} \label{app-function}

Recall that by Theorem~\ref{thm-limits-coincide} we have $d_\theta(x) = 1 - \con(\infty)$.
Let $\pi_1, \pi_2$ be the initial distributions concentrated on $p_1, p_2$, respectively.
It is easy to see that for any word $w \in \Sigma^*$ and any $\mu_1, \mu_2$ with $\mu_1 \le \pi_1^w$ and $\mu_2 \le \pi_2^w$ and $\mu_1 \equiv \mu_2$
 we have $\mu_1(q_1)  = \mu_2(q_1)$ and $\mu_1(q_2)  = \mu_2(q_2)$ and $\mu_1(s) = \mu_2(s) = 0$ for $s \in \{p_1,p_2,r_1,r_2\}$.
Therefore, writing $\mu_1 \land \mu_2$ for the componentwise minimum of row vectors $\mu_1, \mu_2$, we have
$\con(w) = \left| \pi_1^{w} \land \pi_2^{w} \right|$.
Hence $\con(b w) = 0$ holds for all $w \in \Sigma^*$.
So we have for all $k \ge 1$:
\begin{equation}
 \con(k) = \sum_{w \in \Sigma^k} \con(w) = \sum_{w \in \Sigma^{k-1}} \con(a w)
  = \sum_{w \in \Sigma^{k-1}} \left| \pi_1^{a w} \land \pi_2^{a w} \right| \label{eq-app-function1}
\end{equation}
Note that the states $q_1, q_2$ cannot be left after they have been entered.
This motivates the definition of transition matrices restricted to $q_1, q_2$:
\begin{align*}
 A = \begin{pmatrix} \frac12 \quad & \quad  \frac12 - \frac{1}{2 \theta} \\ 0 \quad & \quad \frac{1}{2 \theta} \end{pmatrix} &&&
 B = \begin{pmatrix} \frac{1}{2 \theta} \quad & \quad 0 \\ \frac12 - \frac{1}{2 \theta} \quad & \quad \frac12 \end{pmatrix}
\end{align*}

Define for each $k \ge 1$ a function $\co_k: [0,\infty)^2 \times [0,\infty)^2 \to [0, \infty)$ by
\begin{align*}
 \co_1(u,v)     &= | u \land v | \\
 \co_{k+1}(u,v) &= \co_k(u A, v A) + \co_k(u B, v B)
\end{align*}
By the definition of $\pi_1^{a w}, \pi_2^{a w}$ and by~\eqref{eq-app-function1} we have
\begin{equation}
 \con(k) = \co_k\left(\left(\frac12 - x,0\right),\left(0,\frac12 + x\right)\right) \qquad \text{for all $k \ge 1$.} \label{eq-app-function2}
\end{equation}
Define for each $k \ge 1$ a function $\li_k: \R^2 \to [0, \infty)$ by
\begin{align*}
 \li_1(z)       &= | z | \\
 \li_{k+1}(z)   &= \begin{cases}
                     |z|                     & \text{ if } z \le (0,0) \ \text{ or } \ z \ge (0,0) \\
                     \li_k(z A) + \li_k(z B) & \text{ otherwise,}
                   \end{cases}
\end{align*}
where, for $z_1, z_2 \in \R$, we write $|(z_1,z_2)| = |z_1| + |z_2|$ for the $L_1$-norm.
It follows immediately from the definition that the function $\li_k$ is ``almost linear''
 in the sense that we have
\begin{equation}
 \li_k(a z) = a \li_k(z) = \li_k(-a z)  \qquad \text{for $a \in [0, \infty)$ and $k \in \{1, 2, \ldots\}$.} \label{eq-li-linear}
\end{equation}
The following lemma connects $\co_k$ and~$\li_k$:
\begin{lemma} \label{lem-co-li}
 Let $u \ge (0,0)$ and $v \ge (0,0)$. Then for all $k \ge 1$:
  \[
   2 \co_k(u,v) + \li_k(u-v) = |u + v|\,.
  \]
\end{lemma}
\begin{proof}
We write $u = (u_1, u_2)$ and $v = (v_1, v_2)$.
We proceed by induction on~$k$.
For the induction base let $k = 1$.
Let $u_1 \le v_1$ and $u_2 \le v_2$.
Then we have $\co_1(u,v) = u_1 + u_2$ and $\li_1(u-v) = v_1 - u_1 + v_2 - u_2$.
Hence $2 \co_1(u,v) + \li_1(u-v) = u_1 + v_1 + u_2 + v_2 = |u + v|$.
The case $u_1 \ge v_1$ and $u_2 \ge v_2$ is similar.

Now let $u_1 \le v_1$ and $u_2 \ge v_2$.
Then we have $\co_1(u,v) = u_1 + v_2$ and $\li_1(u-v) = v_1 - u_1 + u_2 - v_2$.
Hence $2 \co_1(u,v) + \li_1(u-v) = u_1 + v_1 + u_2 + v_2 = |u + v|$.
The case $u_1 \ge v_1$ and $u_2 \le v_2$ is similar.

For the induction step let $k \ge 1$.
Then we have:
\begin{equation}
 \li_{k+1}(u-v) = \li_k( (u-v) A ) + \li_k( (u-v) B )\;. \label{eq-co-li-li}
\end{equation}
Indeed, if $u \le v$, we have
\begin{align*}
 & \ \li_{k+1}(u-v) \\
 & = |u-v|                 && \text{(definition of $\li_{k+1}$)}\\
 & = |(u-v) A + (u-v) B|   && \text{($(v-u) \ge 0$ and $A + B$ is stochastic)} \\
 & = |(u-v) A| + |(u-v) B| && \text{($(v-u)$ and $A$ and $B$ are nonnegative)} \\
 & = \li_k( (u-v) A ) + \li_k( (u-v) B ) && \text{(by definition of $\li_k$)}
\end{align*}
The same equalities hold in the case $u \ge v$.
If neither $u \le v$ nor $u \ge v$ holds, then we have $\li_{k+1}(u-v) = \li_k( (u-v) A ) + \li_k( (u-v) B )$ by the definition of~$\li_{k+1}$.

So we have:
\begin{align*}
  & \ 2 \co_{k+1}(u,v) + \li_{k+1}(u-v) \\
  & = 2 \co_k(u A, v A) + 2 \co_k(u B, v B) + \li_{k+1}(u-v) && \text{(definition of $\co_{k+1}$)} \\
  & = 2 \co_k(u A, v A) + 2 \co_k(u B, v B) + \li_k( (u-v) A ) + \li_k( (u-v) B ) && \text{(by~\eqref{eq-co-li-li})} \\
  & = |u A + v A| + |u B + v B| && \text{(induction hypothesis)} \\
  & = |(u + v) (A + B)| && \text{($u,v,A,B$ are nonnegative)} \\
  & = |u + v| && \text{($A+B$ stochastic)}
\end{align*}
\end{proof}

Summarizing the previous development we obtain:
\begin{lemma} \label{lem-dist-li}
We have:
\[
 d_\theta(x) = \frac12 + \frac12 \lim_{k \to \infty} \li_k\left(\left(x - \frac12, x + \frac12 \right)\right)
\]
\end{lemma}
\begin{proof}
\begin{align*}
d_\theta(x)
& = 1 - \con(\infty) && \text{(Theorem~\ref{thm-limits-coincide})} \\
& = 1 - \lim_{k \to \infty} \co_k\left(\left(\frac12 - x,0\right),\left(0,\frac12 + x\right)\right) && \text{(by~\eqref{eq-app-function2})} \\
& = \frac12 + \frac12 \lim_{k \to \infty} \li_k\left(\left(\frac12 - x, -\frac12 - x\right)\right) && \text{(Lemma~\ref{lem-co-li})} \\
& = \frac12 + \frac12 \lim_{k \to \infty} \li_k\left(\left(x - \frac12, x + \frac12\right)\right) && \text{(by~\eqref{eq-li-linear})}
\end{align*}
\end{proof}

Lemma~\ref{lem-dist-li} suggests the definition of a function $\fs{k}: \R \to [0, \infty)$, for each $k \ge 1$, such that
\[
 \fs{k}(x) = \li_k\left(x - \frac12, x + \frac12\right)\,.
\]

The following lemma characterizes $\fs{k}$ recursively:
\begin{lemma} \label{lem-fsk}
 We have $\fs{1}(x) = 2 |x|$.
 Further, for all $k \ge 1$:
 \[
 \fs{k+1}(x) = \begin{cases}
                 2 |x| & |x| \ge \frac12 \\
                 \frac{1}{2 \theta} \fs{k}( \theta x - (\frac12 \theta - \frac12)) +
                 \frac{1}{2 \theta} \fs{k}( \theta x + (\frac12 \theta - \frac12)) & |x| \le \frac12
               \end{cases}
 \]
\end{lemma}
\begin{proof}
The equalities $\fs{1}(x) = 2 |x|$ and $\fs{k+1}(x) = 2 |x|$ for $|x| \ge \frac12$ follow from the definitions.
Further we have:
\begin{align*}
\left(x - \frac12 \ , \ x + \frac12\right) A & \ = \ \frac{1}{2 \theta} \left( \theta x - \frac12 \theta \ , \ \theta x - \frac12 \theta + 1 \right)
 && \text{and} \\
\left(x - \frac12 \ , \ x + \frac12\right) B & \ = \ \frac{1}{2 \theta} \left( \theta x + \frac12 \theta - 1 \ , \ \theta x + \frac12 \theta + 1 \right)
\end{align*}
So it follows for $|x| \le \frac12$:
\begin{align*}
 \fs{k+1}(x) & = \li_{k+1}\left(x - \frac12, x + \frac12\right) && \text{(definition of~$\fs{k+1}$)} \\
             & = \li_k\left( \frac{1}{2 \theta} \left( \theta x - \frac12 \theta \ , \ \theta x - \frac12 \theta + 1 \right) \right)
                 && \text{(definition of $\li_{k+1}$)} \\
             & \quad + \li_k\left( \frac{1}{2 \theta} \left( \theta x + \frac12 \theta - 1 \ , \ \theta x + \frac12 \theta + 1 \right) \right) \\
             & = \frac{1}{2 \theta} \li_k\left( \theta x - \frac12 \theta \ , \ \theta x - \frac12 \theta + 1 \right)
                 && \text{(by~\eqref{eq-li-linear})} \\
             & \quad + \frac{1}{2 \theta} \li_k \left( \theta x + \frac12 \theta - 1 \ , \ \theta x + \frac12 \theta + 1 \right) \\
             & = \frac{1}{2 \theta} \fs{k}\left( \theta x - \left( \frac12 \theta - \frac12 \right) \right)
                 && \text{(definition of~$\fs{k}$)} \\
             & \quad + \frac{1}{2 \theta} \fs{k} \left( \theta x + \left( \frac12 \theta - \frac12 \right) \right)
\end{align*}
\end{proof}

Now we prove Proposition~\ref{prop-distance-function} from the main body of the paper.
\begin{qproposition}{\ref{prop-distance-function}}
\stmtpropdistancefunction
\end{qproposition}
\begin{proof}
We use the Banach fixed-point theorem.
Define a complete metric space $(F, \Delta)$ by
\[
 F := \{ f : \R \to [0, \infty) \mid f \text{ is continuous and } f(x) = 2 |x| \text{ for } |x| \ge {\textstyle \frac12} \}
\]
and the distance metric $\Delta: F \times F \to [0,\infty)$ with
\[
 \Delta(f_1, f_2) := \sup_{|x| \le {\textstyle\frac12}} |f_1(x) - f_2(x)|\;.
\]
Fix $\theta > 1$.
Define the function $S_\theta : F \to F$ with
\[
 S_\theta(f)(x) := \begin{cases} 2 |x| & |x| \ge \frac12 \\
                                \frac{1}{2 \theta} f( \theta x - (\frac12 \theta - \frac12)) +
                                \frac{1}{2 \theta} f( \theta x + (\frac12 \theta - \frac12)) & |x| \le \frac12\,.
                   \end{cases}
\]
We have for all $f_1, f_2 \in F$ and $|x| \le \frac12$:
\begin{align*}
 & |S_\theta(f_1)(x) - S_\theta(f_2)(x)| \\
 & \le \frac{1}{2 \theta} \left| f_1\left(\theta x - \left(\frac12 \theta - \frac12\right)\right)
                               - f_2\left(\theta x - \left(\frac12 \theta - \frac12\right)\right) \right| \\
 & \quad + \frac{1}{2 \theta} \left| f_1\left(\theta x + \left(\frac12 \theta - \frac12\right)\right)
                               - f_2\left(\theta x + \left(\frac12 \theta - \frac12\right)\right) \right| \\
 & \le \frac{1}{2 \theta} \Delta(f_1, f_2) + \frac{1}{2 \theta} \Delta(f_1, f_2) \\
 & = \frac{1}{\theta} \Delta(f_1, f_2)
\end{align*}
It follows that we have $\Delta(S_\theta(f_1), S_\theta(f_2)) \le \frac{1}{\theta} \Delta(f_1, f_2)$,
 so $S_\theta$ is contraction mapping.
Using the Banach fixed-point theorem and Lemma~\ref{lem-fsk} we obtain that
 the function sequence $\fs{1}, \fs{2}, \ldots$ converges to the (unique) fixed point of~$S_\theta$,
  i.e., to the function~$f_\theta$ from the statement of this proposition.
It follows:
\begin{align*}
d_\theta(x) & = \frac12 + \frac12 \lim_{k \to \infty} \li_k\left(\left(x - \frac12, x + \frac12 \right)\right) && \text{(Lemma~\ref{lem-dist-li})} \\
            & = \frac12 + \frac12 \lim_{k \to \infty} \fs{k}(x) && \text{(definition of~$\fs{k}$)} \\
            & = \frac12 + \frac12 f_\theta(x) && \text{(as argued above)}
\end{align*}
\end{proof}

\section{Proof of Lemma~\ref{lem-dist-1-R}} \label{app-distance-1}

We prove Lemma~\ref{lem-dist-1-R} from the main body of the paper.

\begin{qlemma}{\ref{lem-dist-1-R}}
\stmtlemdistoneR
\end{qlemma}
\begin{proof}
Define
 \begin{align*}
  S^{\pi_1, \pi_2} := \{(r_1,r_2) \in Q \times Q \mid \mbox{} & \exists q_1 \in \supp(\pi_1) \ \exists q_2 \in \supp(\pi_2):\\
                                                              & (r_1,r_2) \text{ is reachable from $(q_1, q_2)$ in $G$}\} \,.
 \end{align*}
We need to show $S^{\pi_1, \pi_2} = R^{\pi_1, \pi_2}$.

First we show $S^{\pi_1, \pi_2} \subseteq R^{\pi_1, \pi_2}$.
For $k \in \N$ let
\begin{align*}
 S^{\pi_1, \pi_2}_k := \{(r_1,r_2) \in Q \times Q \mid \mbox{} & \exists q_1 \in \supp(\pi_1) \ \exists q_2 \in \supp(\pi_2):\\
                                                               & (r_1,r_2) \text{ is reachable from $(q_1, q_2)$ in $G$ in $k$ steps.}\}
\end{align*}
We show by induction on~$k$ that
\begin{equation} \label{eq-lem-dist-1-R-1}
 S^{\pi_1, \pi_2}_k \subseteq R^{\pi_1, \pi_2} \text{\quad  for all $k \in \N$.}
\end{equation}
The case $k=0$ is trivial.
Let $k \ge 0$.
Let $(r_1', r_2') \in S^{\pi_1, \pi_2}_{k+1}$.
Then there are $q_1 \in \supp(\pi_1)$ and $q_2 \in \supp(\pi_2)$ and $r_1, r_2 \in Q$ so that there is a path
 \[
  (q_1,q_2) \to \cdots \to (r_1, r_2) \to (r_1', r_2')
 \]
of length~$k+1$ in~$G$.
By the induction hypothesis there is $w \in \Sigma^*$ such that $r_1 \in \supp(\pi_1^w)$ and $r_2 \in \supp(\pi_2^w)$.
By the definition of~$G$ the presence of an edge from $(r_1, r_2)$ to $(r_1', r_2')$ implies that there is
 $a \in \Sigma$ with $M(a)(r_1,r_1') > 0$ and $M(a)(r_2,r_2') > 0$.
It follows that we have $r_1' \in \supp(\pi_1^{w a})$ and $r_2' \in \supp(\pi_2^{w a})$.
Hence $(r_1', r_2') \in R^{\pi_1, \pi_2}$.
This proves~\eqref{eq-lem-dist-1-R-1}.
Since $\bigcup_{k \ge 0} S^{\pi_1, \pi_2}_k = S^{\pi_1, \pi_2}$, we have also shown $S^{\pi_1, \pi_2} \subseteq R^{\pi_1, \pi_2}$.

Now we show $R^{\pi_1, \pi_2} \subseteq S^{\pi_1, \pi_2}$.
For $w \in \Sigma^*$ let
 \[
  R^{\pi_1, \pi_2}_w := \{(r_1,r_2) \in Q \times Q \mid r_1 \in \supp(\pi_1^w) \text{ and } r_2 \in \supp(\pi_2^w)\}\,.
 \]
We show by induction on the length of~$w$ that
\begin{equation} \label{eq-lem-dist-1-R-2}
 R^{\pi_1, \pi_2}_w \subseteq S^{\pi_1, \pi_2} \text{\quad  for all $w \in \Sigma^*$.}
\end{equation}
The case $|w| = 0$ is trivial.
Let $w \in \Sigma^*$ and $a \in \Sigma$.
Let $(r_1', r_2') \in R^{\pi_1, \pi_2}_{w a}$, i.e., $r_i' \in \supp(\pi_i^w M(a))$ for $i \in \{1,2\}$.
Hence there are $r_1, r_2 \in Q$ with $r_i \in \supp(\pi_i^w)$ and $M(a)(r_i, r_i') > 0$ for $i \in \{1,2\}$.
By the induction hypotheses we have that $(r_1,r_2)$ is reachable from $(q_1, q_2)$ in $G$.
By the definition of~$G$ there is an edge from $(r_1,r_2)$ to $(r_1',r_2')$ in~$G$.
Hence $(r_1',r_2')$ is reachable from $(q_1, q_2)$ in $G$, so $(r_1',r_2') \in S^{\pi_1, \pi_2}$ and \eqref{eq-lem-dist-1-R-2} is proved.
Since $\bigcup_{w \in \Sigma^*} R^{\pi_1, \pi_2}_w = R^{\pi_1, \pi_2}$, we have also shown $R^{\pi_1, \pi_2} \subseteq S^{\pi_1, \pi_2}$.
\end{proof}

}{}

%
%
%


\end{document}